\documentclass[12pt]{iopart}
\usepackage{iopams}
\usepackage{amsthm}
\usepackage{dsfont}
\usepackage{color}
\usepackage{ulem}
\usepackage{verbatim}

\usepackage[pdfpagemode=None]{hyperref} 
\hypersetup{
    colorlinks,
    citecolor=black,
    filecolor=black,
    linkcolor=black,
    urlcolor=blue
}
\newcommand {\link}{\href}

\newcommand {\be}{\begin{equation}}
\newcommand {\ee}{\end{equation}}

\newtheorem{proposition}{Proposition}
\renewenvironment{proof}{\noindent{\bf Proof.}}{{\hspace*{\fill} $\square$}\newline}

\theoremstyle{definition}

\newtheorem{lemma}{Lemma}
\theoremstyle{remark}
\newtheorem{remark}{Remark}
\makeatletter
\@addtoreset{remark}{utv}
\@addtoreset{remark}{theorem}
\makeatother

\renewcommand{\Im}{\mathop{\mathrm{Im}}\nolimits}
\newcommand{\Ker}{\mathop{\mathrm{Ker}}\nolimits}
\newcommand{\diag}{\mathop{\mathrm{diag}}\nolimits}

\begin{document}

\title{Quasibosons composed of two $q$-fermions: realization by deformed oscillators}

\author{A.M.~Gavrilik, I.I.~Kachurik and Yu.A. Mishchenko}

\address{Bogolyubov Institute for Theoretical Physics, Kiev 03680, Ukraine}
\ead{omgavr@bitp.kiev.ua}

\vspace{4mm}

\vspace{2mm}

\begin{abstract}
Composite bosons, here called {\it quasibosons} (e.g. mesons, excitons, etc.) occur in various physical situations. Quasibosons differ from bosons or fermions as their creation and annihilation operators obey non-standard commutation relations,  even for the ``fermion+fermion''  composites. Our aim is to realize the operator algebra of quasibosons composed of two fermions or two $q$-fermions ($q$-deformed fermions) by the respective operators of
deformed oscillators, the widely studied objects. For this, the restrictions on quasiboson creation/annihilation operators and on the deformed oscillator (deformed boson) algebra are obtained. Their resolving proves the uniqueness of the family of deformations and gives explicitly the deformation structure function (DSF) which provides the desired realization. In the case of two fermions as con\-sti\-tu\-ents, such realization is achieved when the DSF is quadratic polynomial in the num\-ber operator. In the case of two $q$-fermions, $q\!\ne\! 1$, the obtained DSF inherits the pa\-ra\-me\-ter $q$ and does not continuously converge when $q\!\to\! 1$ to the DSF of the first case.
\end{abstract}


\section{Introduction}

Theoretical treatment of many-particle systems is connected with a number of comp\-li\-ca\-ti\-ons. Some of them can be resolved by introducing the concept of quasiparticle or ``composite particle'', if this is possible. However, in this way we generally en\-co\-un\-ter various factors of the internal structure, which cannot be completely encapsulated into internal degrees of freedom of a composite particle. These are the nontrivial com\-mu\-ta\-ti\-on relations, or the interaction of the constituents between themselves and with other particles, etc. It is desirable to have an equivalent description of many-(composite-)par\-ti\-cle systems, almost as simple as the description of an ideal/point-like particle system, but taking into account the mentioned factors. Deformed bosons or deformed oscillators, see e.g. the review \cite{Bona}, provide possible means for the realization of such an intention. In such a case the basic characteristics of the factors connected with the internal structure would be encoded in one or more deformation parameters.

A particular realization of the mentioned idea to describe quasibosons~\cite{Per} (boson-like composites) in terms of deformed Heisenberg algebra was demonstrated by Avancini and Krein in \cite{Avan} who utilized the quonic \cite{Green} version of the deformed boson al\-geb\-ra. Note that if two or more copies (modes) are involved, different modes of quons do not commute \cite{Avan,Green}. Unlike quons, the deformed oscillators of Arik-Coon type are in\-de\-pen\-dent \cite{Arik,Jag-etal}, that is, the operators corresponding to different copies, mutually commute.

Regardless of their intrinsic origin and physical motivation, diverse models of deformed oscillators have received much attention during the 1990s and till now. Among the best known and extensively studied deformed oscillators one encounters the $q$-deformed Arik-Coon (AC) \cite{Arik}  or Biedenharn-Macfarlane (BM) \cite{Biedenharn,Macfarlane} ones, the $q$-deformed Tamm-Dancoff oscillator~\cite{TD,TD2,GR1}, and also the two-parameter $p,\!q$-deformed oscillator \cite{p-q,Arik-Fibo}. On the other hand, the so-called $\mu$-deformed oscillator is much less studied. Introduced in \cite{Jan} almost two decades ago, this deformed oscillator essentially differs from the models we have already mentioned and exhibits rather unusual properties~\cite{GKR,GR2}. Note that there exists a general approach to the description of deformed oscillators based on the concept of the deformation structure function (DSF) given in~\cite{Melja,Bona}. As the extension of the standard quantum harmonic oscillator, deformed oscillators find diverse applications in describing miscellaneous physical systems involving essential nonlinearities, from say quantum optics and the Landau problem to high energy particle phenomenology and modern quantum field theory, see e.g. \cite{Man'ko,Rego2,Rego3,AGI1,AGI1-2,AGI2,Avan2003,AG,gavrS,Rego1}.

Although a great variety of models of deformed oscillators exists as mentioned above, the detailed analysis of possible realizations, on their base, of composite particles along with the interpretation of deformation parameters in terms of the internal structure as far as we know is lacking. To fill this gap, in our preceding paper \cite{GKM-1} some steps in that direction were undertaken and first results were obtained.  Namely, we carried out the detailed analysis for quasibosons consisting of two ordinary fermions with the ansatz $A^{\dag}_{\alpha}=\Phi^{\mu\nu}_{\alpha}a^{\dag}_{\mu}b^{\dag}_{\nu}$ for the quasiboson creation operator in the $\alpha$-th mode, meaning the bilinear combination of the constituents' creation operators of the general form. The analysis implies the realization of quasibosons by deformed oscillators characterized by the most general DSF $\phi(N)$ which unambiguously determines \cite{Melja,Bona} the deformed algebra within one mode. Our present study further extends the results obtained in \cite{GKM-1} by using, instead of the usual fermions,  {\it their $q$-deformed analog} for the constituents' operators.

The paper is organized as follows. Section~\ref{sec2}, which serves as the base for our analysis, concerns the case of quasibosons whose constituents are ordinary fermions (the particular $q=1$ case of $q$-fermions). Here, after introducing the creation and annihilation operators for composite quasibosons, we recapitulate main facts and results from \cite{GKM-1} (note that some of these results, only sketched in \cite{GKM-1}, here are presented in full detail: in particular, that concerns the extended treatment given in subsection~\ref{sec5}). We establish important relations for
quasibosons' operators that include necessary conditions for the representation of quasibosons in terms of deformed bosons to hold. Those conditions are partially solved in subsection~\ref{sec4} yielding the DSFs~$\phi(N)$ of the effective deformation, and completely solved in subsection~\ref{sec5}. There we obtain explicitly all possible internal structures for quasibosons with the corresponding matrices $\Phi_{\alpha}^{\mu\nu}$. In Section~\ref{sec7}, presenting the further development of the ideas and results of \cite{GKM-1}, for the constituents' operators we take instead of usual fermions their $q$-deformed analogs. The corresponding treatment is performed: the admissible (for the realization under question) structure function~$\phi(n)$ and matrices $\Phi_{\alpha}^{\mu\nu}$ are found as the solution of the necessary conditions for the validity of the realization. Simpler illustrative examples, along with intermediate proofs, are relegated to appendices. The paper ends with concluding remarks and some outlook.

\section{System of quasibosons composed of two fermions}\label{sec2}

The general task of representing the quasibosons consisting of $q$-fermions can be divided into two particular situations: i) the constituents are pure fermions ($q=1$); ii) the constituents are essentially deformed $q$-fermions ($q\neq 1$). This section is devoted to the first case: similar to~\cite{GKM-1} we deal with the system of composite boson-like particles ({\it quasibosons} \cite{Per}) such that each copy (mode) of them is built from two fermions. We study the realization of quasibosons in terms of the set of {\it independent} identical copies of deformed oscillators of the general form (for some examples of mode-independent systems  see~\cite{Jag-etal}).

Let us denote the creation and annihilation operators of the two (mutually anticommuting) sets of usual fermions by $a^{\dag}_{\mu}$,
$b^{\dag}_{\nu}$, $a_{\mu}$, $b_{\nu}$ respectively, with their standard anticommutation relations, namely
\begin{equation}
\eqalign{
\{a_{\mu},a^{\dag}_{\mu'}\}\equiv
a_{\mu}a^{\dag}_{\mu'}+a^{\dag}_{\mu'}a_{\mu}=\delta_{\mu\mu'},\qquad&\{a_{\mu},a_{\nu}\}=0,\cr
\{b_{\nu},b^{\dag}_{\nu'}\}\equiv
b_{\nu}b^{\dag}_{\nu'}+b^{\dag}_{\nu'}b_{\nu}=\delta_{\nu\nu'},\qquad&\{b_{\mu},b_{\nu}\}=0.
}
\end{equation}
Besides, each of $a^{\dag}_{\mu}$, $a_{\mu}$ anticommutes  with each of $b^{\dag}_{\nu}$,  $b_{\nu}$.  So, we use these fermions to construct quasibosons. Then, the corresponding quasibosonic creation and annihilation operators $A^{\dag}_{\alpha},\ A_{\alpha}$ (where $\alpha$ labels the particular quasiboson and denotes the whole set of its quantum numbers) are given as
\begin{equation} \label{anzats}
A^{\dag}_{\alpha}=\sum\limits_{\mu\nu}\Phi^{\mu\nu}_{\alpha}a^{\dag}_{\mu}b^{\dag}_{\nu},\quad
A_{\alpha}=\sum\limits_{\mu\nu}\overline{\Phi}^{\mu\nu}_{\alpha}b_{\nu}a_{\mu} \,.
\end{equation}
For the matrices $\Phi_{\alpha}$ we assume the following normalization condition:
\[
\sum\limits_{\mu\nu}\Phi^{\mu\nu}_{\alpha}\overline{\Phi}^{\mu\nu}_{\beta}
\  \equiv \Tr \Phi_{\alpha}\Phi^{\dag}_{\beta}=\delta_{\alpha\beta}
 \,.
\]
One can easily check that
\begin{equation}
[A_{\alpha},A_{\beta}]=[A^{\dag}_{\alpha},A^{\dag}_{\beta}]=0 .
\label{2_2}
\end{equation}
For the remaining commutator one finds~\cite{Avan}
\begin{equation}    \label{2_3}
[A_{\alpha},A^{\dag}_{\beta}] =
\sum\limits_{\mu\nu\mu'\nu'}\overline{\Phi}^{\mu\nu}_{\alpha}\Phi^{\mu'\nu'}_{\beta}
\left([a_{\mu},a^{\dag}_{\mu'}]b_{\nu}b^{\dag}_{\nu'} +
a^{\dag}_{\mu'}a_{\mu} [b_{\nu},b^{\dag}_{\nu'}]\right)
  = \delta_{\alpha\beta} - \Delta_{\alpha\beta}
\end{equation}
where
\begin{equation*}
\Delta_{\alpha\beta} \equiv
\sum\limits_{\mu\nu\mu'}\overline{\Phi}^{\mu\nu}_{\alpha}
\Phi^{\mu'\nu}_{\beta}a^{\dag}_{\mu'}a_{\mu}
+\sum\limits_{\mu\nu\nu'}\overline{\Phi}^{\mu\nu}_{\alpha}
\Phi^{\mu\nu'}_{\beta}b^{\dag}_{\nu'}b_{\nu}.
\end{equation*}
The entity $\Delta_{\alpha\beta}$ in (\ref{2_3}) shows  deviation from the pure bosonic canonical relation. Note that if $\Delta_{\alpha\beta}=0$ then we have ${\Phi}^{\mu\nu}_{\alpha}=0$.

Remark that unlike the realization of quasibosonic operators using the quonic version of the deformed oscillator algebra, as was done in~\cite{Avan}, in all our analysis we consider (the set of) completely independent copies of deformed oscillators. That is, we assume the validity of (\ref{2_2}) and also require $[A_{\alpha},A^{\dag}_{\beta}]=0$ for $\alpha\neq\beta$.

The most simple type of deformed oscillator is the Arik-Coon $q$-deformation \cite{Arik}.  So it is of interest, first, to try to use this set of $q$-deformed bosons for representing the system of independent quasibosons. However, as was shown in \cite{GKM-1}, the representation of quasibosons with the {\it independent}  system of $q$-deformed bosons of the Arik-Coon type leads to inconsistency. For that reason we set the goal to examine other deformed oscillators in the general form given by their structure function~$\phi(N)$.

\paragraph{Necessary conventions.}

Our goal is to operate with $A_{\alpha}$, $A^{\dag}_{\alpha}$ and $N_{\alpha}$ constructed from $a^{\dag}_{\mu},a_{\mu},b^{\dag}_{\nu},b_{\nu}$ ($N_{\alpha}$ is some effective number operator for composite particles) as with the elements (operators) of some deformed oscillator algebra,
``forgetting'' about their internal structure. It means that we are looking for subalgebras of the enveloping algebra
$\mathfrak{A}\{A_{\alpha},A^{\dag}_{\alpha},N_{\alpha}\}$, generated by $A_{\alpha}$, $A^{\dag}_{\alpha}$, $N_{\alpha}$, isomorphic to
some deformed oscillator algebras $\mathfrak{A}\{\mathcal{A}_{\alpha},\mathcal{A}^{\dag}_{\alpha},\mathcal{N}_{\alpha}\}$, generated by $\mathcal{A}_{\alpha}$, $\mathcal{A}^{\dag}_{\alpha}$, $\mathcal{N}_{\alpha}$:
\[
\mathfrak{A}\{A_{\alpha},A^{\dag}_{\alpha},N_{\alpha}\}\simeq
\mathfrak{A}\{\mathcal{A}_{\alpha},\mathcal{A}^{\dag}_{\alpha},\mathcal{N}_{\alpha}\}.
\]
We will establish necessary and sufficient conditions for the existence of such isomorphism. We also require the isomorphism of representation spaces of the mentioned algebras:
\begin{equation}
L\{(a^{\dag}_{\mu})^r(b^{\dag}_{\nu})^s\ldots|O\rangle\}\supset H\simeq
\mathcal{H}=L\{\mathcal{A}^{\dag}_{\gamma_1}\ldots\mathcal{A}^{\dag}_{\gamma_n}|O\rangle\},
\end{equation}
where $L\{...\}$ denotes a linear span.  Thus, if the algebra of deformed oscillator operators is given by the relations
\begin{eqnarray}
G_i(\mathcal{A}_{\alpha},\mathcal{A}^{\dag}_{\alpha},\mathcal{N}_{\alpha})=0
\quad\Leftrightarrow\quad
G_i(\mathcal{A}_{\alpha},\mathcal{A}^{\dag}_{\alpha},
\mathcal{N}_{\alpha})\mathcal{A}^{\dag}_{\gamma_1}\ldots
\mathcal{A}^{\dag}_{\gamma_n}|O\rangle=0, \\
\nonumber  \hspace{60mm} n=0,1,2,...\label{rels}
\end{eqnarray}
then necessary and sufficient conditions for  the isomorphism to exist can be written as
\begin{equation}\label{isom_cond}
G_i(A_{\alpha},A^{\dag}_{\alpha},N_{\alpha})\cong0
\quad\mathop{\Longleftrightarrow}^{def}\quad
G_i(A_{\alpha},A^{\dag}_{\alpha},N_{\alpha})
A^{\dag}_{\gamma_1}\ldots A^{\dag}_{\gamma_n}|O\rangle=0.
\end{equation}
Here the symbol of the weak equality $\cong$ is introduced which means the equality on all the $n$-(quasi)boson states.  Next, we observe that
\begin{equation}\nonumber
G_i A^{\dag}_{\gamma_1}|O\rangle=0 \quad\Leftrightarrow\quad
[G_i,A^{\dag}_{\gamma_1}]|O\rangle=0
\end{equation}
and, by induction,
\begin{equation} \nonumber
G_i A^{\dag}_{\gamma_1}...A^{\dag}_{\gamma_n}|O\rangle=0
\quad\Leftrightarrow\quad
[...[G_i,A^{\dag}_{\gamma_1}]...,A^{\dag}_{\gamma_n}]|O\rangle=0.
\end{equation}
For a general deformed oscillator  defined using the structure function $\phi(N)$, see e.g. \cite{Bona}, relation (\ref{rels}) takes the form
\begin{equation}\label{system1_0}
\left\{ \eqalign{ \mathcal{A}^{\dag}_{\alpha}\mathcal{A}_{\alpha} =
\phi(\mathcal{N}_{\alpha}),\cr
[\mathcal{A}_{\alpha},\mathcal{A}^{\dag}_{\alpha}]  =
\phi(\mathcal{N}_{\alpha}+1)-\phi(\mathcal{N}_{\alpha}),\cr
[\mathcal{N}_{\alpha},\mathcal{A}^{\dag}_{\alpha}] =
\mathcal{A}^{\dag}_{\alpha},\, \ \ \
[\mathcal{N}_{\alpha},\mathcal{A}_{\alpha}] = -\mathcal{A}_{\alpha}.
} \right.
\end{equation}
Here the expressions  for $[\mathcal{A}_{\alpha},\mathcal{A}^{\dag}_{\beta}],\ \alpha\neq\beta$, if any, may be added. Thus, the set of functions $G_i$ applicable in this case reads as follows:
\begin{eqnarray*}
G_0(A_{\alpha},A^{\dag}_{\alpha},N_{\alpha}) =
A^{\dag}_{\alpha}A_{\alpha} -
\phi(N_{\alpha}),\\
 G_1(A_{\alpha},A^{\dag}_{\alpha},N_{\alpha}) =
 [A_{\alpha},A^{\dag}_{\alpha}] - \bigl(\phi(N_{\alpha}+1)-\phi(N_{\alpha})\bigr),\\
G_2(A^{\dag}_{\alpha},N_{\alpha}) =
[N_{\alpha},A^{\dag}_{\alpha}] - A^{\dag}_{\alpha}, \quad {\rm and\ possibly\ some\ others.}
\end{eqnarray*}
Such functions $G_i$ are determined  by the structure function of deformation $\phi(N_{\alpha})$. So, relations (\ref{isom_cond}) can be used for deducing the connection between matrices $\Phi^{\mu\nu}_{\alpha}$, which determine the operators $A^{\dag}_{\alpha}$, and the DSF $\phi(N_{\alpha})$.

\subsection{Necessary conditions on $\Phi_{\alpha}^{\mu\nu}$ and
$\phi(n)$}\label{sec4}

In the subsequent analysis we study  the independent quasibosons' system realized by deformed oscillators without an indication of the
particular model of deformation. The aim of this section is to obtain necessary conditions for such realization in terms of the matrices $\Phi_{\alpha}$. Note that the results of this section are not sensitive to the form of the definition of $N_{\alpha}(\cdot)$ as a function of $A_{\alpha}$, $A^{\dag}_{\alpha}$.

Using relations (\ref{isom_cond})-(\ref{system1_0}) and taking into account the independence of modes, we arrive at the following weak equalities for the commutators:
\begin{equation}
\left\{
\eqalign{
[A_{\alpha},A^{\dag}_{\beta}]\cong 0 \quad{\rm for}\quad \alpha\neq\beta,\cr
[N_{\alpha},A^{\dag}_{\alpha}]\cong A^{\dag}_{\alpha},\quad [N_{\alpha},A_{\alpha}]\cong -A_{\alpha},\cr
[A_{\alpha},A^{\dag}_{\alpha}] \cong \phi(N_{\alpha}+1)-\phi(N_{\alpha}).
}
\right.\label{system1}
\end{equation}

\paragraph{Treatment of mode independence.}

From the first relation in (\ref{system1})  we derive the equivalent requirements of independence in terms of matrices $\Phi$:
\begin{equation}
\sum\limits_{\mu'\nu'}\left(\Phi^{\mu\nu'}_{\beta}
\overline{\Phi}^{\mu'\nu'}_{\alpha}\Phi^{\mu'\nu}_{\gamma}
+ \Phi^{\mu\nu'}_{\gamma}\overline{\Phi}^{\mu'\nu'}_{\alpha}\Phi^{\mu'\nu}_{\beta}\right)
=0,\quad \alpha\neq\beta,\label{uslnez}
\end{equation}
which can be rewritten in the matrix form
\begin{equation}
\Phi_{\beta}\Phi^{\dag}_{\alpha}\Phi_{\gamma}+
\Phi_{\gamma}\Phi^{\dag}_{\alpha}\Phi_{\beta}=0,\quad
\alpha\neq\beta.\label{req1}
\end{equation}

\paragraph{Conditions on $\Phi_{\alpha}^{\mu\nu}$ within one mode $\alpha$.}

Since $A^{\dag}_{\alpha}A_{\alpha}\cong\phi(N_{\alpha})$ and $A_{\alpha}A^{\dag}_{\alpha}\cong\phi(N_{\alpha}\!+1)$, we have
\begin{equation}  \nonumber
[A^{\dag}_{\alpha}A_{\alpha},A_{\alpha}A^{\dag}_{\alpha}]\cong 0 \qquad {\rm and} \qquad
[\Delta_{\alpha\alpha},N_{\alpha}]\cong 0.
\end{equation}
The first equality can equivalently be rewritten as
\[
[A^{\dag}_{\alpha}A_{\alpha},\Delta_{\alpha\alpha}]= [A^{\dag}_{\alpha}A_{\alpha},\sum\limits_{\mu\nu\mu'}\overline{\Phi}^{\mu\nu}_{\alpha}
\Phi^{\mu'\nu}_{\alpha}a^{\dag}_{\mu'}a_{\mu} +\sum\limits_{\mu\nu\nu'}\overline{\Phi}^{\mu\nu}_{\alpha}
\Phi^{\mu\nu'}_{\alpha}b^{\dag}_{\nu'}b_{\nu}]\cong 0.
\]
The calculation of this commutator gives
\begin{eqnarray}
[A^{\dag}_{\alpha}A_{\alpha},\Delta_{\alpha\alpha}]=
2A^{\dag}_{\alpha}\sum\limits_{\mu\nu}\left(\Psi_{\alpha}^{\dag}\right)^{\nu\mu}b_{\nu}a_{\mu}
-2\sum\limits_{\mu'\nu'}\Psi_{\alpha}^{\mu'\nu'}a^{\dag}_{\mu'}b^{\dag}_{\nu'}
A_{\alpha}\cong 0,   \label{comm2} \\
\nonumber
\Psi_{\alpha} \equiv \Phi_{\alpha}\Phi^{\dag}_{\alpha}\Phi_{\alpha}.
\end{eqnarray}
With the account of (\ref{anzats}) one can see: the validity of (\ref{comm2}) on the one-quasiboson state requires that the commutator with  the creation operator on the vacuum should be
\begin{eqnarray*}
\left[(\overline{\Psi}^{\mu\nu}_{\alpha}\Phi^{\mu'\nu'}_{\alpha} -
\overline{\Phi}^{\mu\nu}_{\alpha}\Psi^{\mu'\nu'}_{\alpha})
a^{\dag}_{\mu'}b^{\dag}_{\nu'}b_{\nu}a_{\mu},\Phi_{\alpha}^{\lambda\rho}
a^{\dag}_{\lambda}b^{\dag}_{\rho}\right]|O\rangle= \\
=\!\left(\!\Phi_{\alpha}^{\mu'\nu'}\!a^{\dag}_{\mu'}b^{\dag}_{\nu'} \!\cdot\!
\overline{\Psi}_{\alpha}^{\mu\nu}\Phi_{\alpha}^{\mu\nu}  \!-\!
\Phi_{\alpha}^{\mu'\nu'}\!a^{\dag}_{\mu'}b^{\dag}_{\nu'}
\Delta[\Psi,\Phi] \!-\!\Psi_{\alpha}^{\mu'\nu'}\!a^{\dag}_{\mu'}b^{\dag}_{\nu'} \!+\!
\Psi_{\alpha}^{\mu'\nu'}\!a^{\dag}_{\mu'}b^{\dag}_{\nu'} \!\cdot\!
\Delta_{\alpha\alpha}\!\right)\!|O\rangle \\
=\left(\Phi_{\alpha}^{\mu'\nu'}\cdot
\Tr(\Psi_{\alpha}^{\dag}\Phi_{\alpha}) -
\Psi_{\alpha}^{\mu'\nu'}\right)a^{\dag}_{\mu'}b^{\dag}_{\nu'}|O\rangle=0
\end{eqnarray*}
(the summation over repeated indices is meant). From this we obtain the requirement
\begin{equation}
\Phi_{\alpha}\Phi^{\dag}_{\alpha}\Phi_{\alpha} =
\Tr(\Phi^{\dag}_{\alpha}\Phi_{\alpha}\Phi^{\dag}_{\alpha}\Phi_{\alpha})\cdot
\Phi_{\alpha}\, , \label{req2}
\end{equation}
which is also the sufficient one. This requirement guarantees not only the weak equality as in (\ref{comm2}) but also the corresponding strong (operator) equality.

Thus, we have two independent  requirements (\ref{req1}) and (\ref{req2}) for the matrices $\Phi_{\alpha}.$

\paragraph{Relating $\Phi_{\alpha}$ to the structure function $\phi(n)$.}

Let us derive the relations that involve the DSF $\phi$. Directly from the system (\ref{system1}) we obtain the initial values for the DSF $\phi$:
\begin{eqnarray*}
\phi(N_{\alpha})\cong A^{\dag}_{\alpha}A_{\alpha} \quad\quad\ \ \Rightarrow\quad\phi(0)=0,\qquad\qquad\\
\phi(N_{\alpha}+1)\cong A_{\alpha}A^{\dag}_{\alpha} \quad\Rightarrow\quad\phi(1)=1.\qquad\qquad
\end{eqnarray*}
From (\ref{2_3}) and the third relation in (\ref{system1}) we have
\begin{equation*}
[A_{\alpha},A^{\dag}_{\alpha}] = 1-\Delta_{\alpha\alpha} \cong
\phi(N_{\alpha}+1)-\phi(N_{\alpha}),
\end{equation*}
or, equivalently,
\begin{equation*}
F_{\alpha\alpha} \equiv \Delta_{\alpha\alpha} - 1 +
\phi(N_{\alpha}+1)-\phi(N_{\alpha}) \cong 0.
\end{equation*}
If the conditions (see (\ref{system1}))
\begin{equation}
[N_{\alpha},A^{\dag}_{\alpha}]\cong
A^{\dag}_{\alpha},\quad [N_{\alpha},A_{\alpha}]\cong
-A_{\alpha}\label{n_commut}
\end{equation}
do hold (it means that for these relations a subsequent verification is needed), then
\begin{equation*}
\phi(N_{\alpha})A^{\dag}_{\alpha}\!\cong\! A^{\dag}_{\alpha}
\phi(N_{\alpha}\!+\!1)\ \ \Rightarrow \ \
[\phi(N_{\alpha}),A^{\dag}_{\alpha}]\!\cong\! A^{\dag}_{\alpha}
\bigl(\phi(N_{\alpha}\!+\!1)\!-\!\phi(N_{\alpha})\bigr).
\end{equation*}
As a result, we come to
\begin{equation}\label{5_5}
[F_{\alpha\alpha},A^{\dag}_{\alpha}] \!\cong\!
2(\Phi_{\alpha}\Phi^{\dag}_{\alpha}
 \Phi_{\alpha})^{\mu\nu}a^{\dag}_{\mu}b^{\dag}_{\nu}\!+\!A^{\dag}_{\alpha}
 \Bigl(\phi(N_{\alpha}\!+\!2)\!-\!2\phi(N_{\alpha}\!+\!1)\!+\!
\phi(N_{\alpha})\Bigr).
\end{equation}
Requiring that this commutator vanishes  on the vacuum and taking into account that $\phi(0)=0$, $\phi(1)=1$ we obtain
\begin{equation*}
\Phi_{\alpha}\Phi^{\dag}_{\alpha}\Phi_{\alpha} =  \Bigl(1-\frac12
\phi(2)\Bigr)\Phi_{\alpha} = \frac{f}{2} \Phi_{\alpha}
\end{equation*}
where the deformation parameter $f$ does appear:
\[
\frac{f}{2}\equiv 1-\frac12 \phi(2)=
\Tr(\Phi^{\dag}_{\alpha}\Phi_{\alpha}\Phi^{\dag}_{\alpha}\Phi_{\alpha})
\ \ \ {\rm for\ all}\  \alpha.
\]

\paragraph{Finding admissible $\phi(n)$ explicitly.}

Equality (\ref{5_5}) can be rewritten as
\begin{equation*}
[F_{\alpha\alpha},A^{\dag}_{\alpha}] \cong
\bigl(2-\phi(2)\bigr)A^{\dag}_{\alpha} + A^{\dag}_{\alpha}
\bigl(\phi(N_{\alpha}+2) - 2\phi(N_{\alpha}+1) +
\phi(N_{\alpha})\bigr).
\end{equation*}
By induction, the equality for the $n$-th commutator can be proven:
\begin{equation*}
[\ldots[F_{\alpha\alpha},A^{\dag}_{\alpha}]\ldots A^{\dag}_{\alpha}]
\cong (A^{\dag}_{\alpha})^n \biggl\{\sum\limits_{k=0}^{n+1}
(-1)^{n+1-k} C^k_{n+1}\phi(N_{\alpha}+k) \biggr\}
\end{equation*}
(here  $C_n^k$ denotes binomial coefficients). The requirement that the $n$-th commutator vanishes on the vacuum leads to the recurrence relation
\begin{equation}
\phi(n+1) = \sum\limits_{k=0}^{n} (-1)^{n-k} C^k_{n+1}\phi(k), \quad n\geq
2.\label{recurr1}
\end{equation}
As can be seen, all the values $\phi(n)$ for $n\geq 3$ are  determined unambiguously by the two values $\phi(1)$ and $\phi(2)$, which may in general depend on one or more deformation parameters. Taking into account the equality \cite{Korn}
\begin{equation*}
\sum\limits_{k=0}^{n} (-1)^{n-k} k^m C_n^k  = \left\{
\eqalign{
0,\quad m<n,\cr
n!,\quad m = n,
}
\right.
\end{equation*}
we find: the only independent solutions of (\ref{recurr1}) are $n$ and $n^2$, as well as their linear combination
\begin{equation}
\phi(n)=\left(1+\frac{f}{2}\right)n - \frac{f}{2}n^2.\label{solution1}
\end{equation}
This structure function satisfies both the initial conditions and the recurrence relations in (\ref{recurr1}).
\begin{remark}\label{rem1}
In view of the uniqueness of the solution with fixed initial conditions, formula (\ref{solution1}) gives the general solution of
(\ref{recurr1}).
\end{remark}
\begin{remark}\label{rem2}
If we take the Hamiltonian in the form $H=\frac12 \bigl(\phi(N)+\phi(N+1)\bigr)$ then using the obtained results it is not difficult to derive the three-term recurrence relations for both the deformation structure function and energy eigenspectrum:
\begin{eqnarray*}
\phi(n+1)=\frac{2(n+1)}{n}\phi(n)-\frac{n+1}{n-1}\phi(n-1),\\
E_{n+1}=\frac{4n^2+4n-4}{2n^2-1}E_{n}-\frac{2n^2+4n+1}{2n^2-1}E_{n-1}.
\end{eqnarray*}
The latter equality has a typical form of the so-called quasi-Fibonacci \cite{GKR} relation for the eigenenergies. Note that the general case of deformed oscillators with polynomial structure functions $\phi(N)$ (these are quasi-Fibonacci as well) was studied in \cite{GR5}.
\end{remark}

\subsection{Treatment of the quasiboson number operator}

The quasiboson number operator $N_{\alpha}$ can be introduced in different ways. Its definition is dictated by the requirements $G_0\cong 0$, $G_1\cong 0$ (recall that $G_0$ and $G_1$ are defined just after (\ref{system1_0})) and also by the self-consistency of the realization. A possible definition could be given by the relation $N_{\alpha}\mathop{=}\limits^{def} \! \phi^{-1}\!(A^{\dag}_{\alpha}A_{\alpha})$, or by $N_{\alpha}\mathop{=}\limits^{def}\!\phi^{-1}\!(A_{\alpha}A^{\dag}_{\alpha})-\!1$. We will not choose some of the two forms of definition, but consider the general possibility:
\begin{equation*}
N_{\alpha}\mathop{=}\limits^{def} \chi(A^{\dag}_{\alpha}A_{\alpha},\varepsilon_{\alpha}),\quad {\rm where}\quad \varepsilon_{\alpha}\equiv 1-\Delta_{\alpha\alpha}= [A_{\alpha},A^{\dag}_{\alpha}].
\end{equation*}

As we have mentioned above, it remains to satisfy relations (\ref{n_commut}), which enable to define the function $\chi$. Note that the second of them stems by conjugation from the first one,
\begin{equation} \label{usl2}
[N_{\alpha},A^{\dag}_{\alpha}]\cong
A^{\dag}_{\alpha}.
\end{equation}
Since we assume the {\it independence} of different modes, see (\ref{system1}), we consider the case $\gamma_1=\gamma_2=\ldots=\alpha$ in the definition~(\ref{isom_cond}).

It is useful to denote by $L_n$ the operators
\begin{equation}
L_0=N,\quad L_{n+1}=[L_n,A^{\dag}_{\alpha}]=
[\ldots[N_{\alpha},A^{\dag}_{\alpha}]\ldots A^{\dag}_{\alpha}], \ \
\ n\geq 0 \,. \label{5_22}
\end{equation}
Taking this into account, condition (\ref{usl2}) can be written as
\begin{equation}
L_1|O\rangle=A^{\dag}_{\alpha}|O\rangle,\quad
L_n|O\rangle=0,\quad n>1.\label{q3}
\end{equation}
Now consider three useful statements.
\begin{proposition}\label{prop1}
The following relations are true:
\begin{eqnarray*}
[\Delta_{\alpha\alpha},A^{\dag}_{\alpha}]=f A^{\dag}_{\alpha},\quad [\Delta_{\alpha\alpha},A_{\alpha}]=
-\overline{f}A_{\alpha},\quad f=
2\Tr(\Phi^{\dag}_{\alpha}\Phi_{\alpha}\Phi^{\dag}_{\alpha}\Phi_{\alpha}),\cr
[\varepsilon_{\alpha},A^{\dag}_{\alpha}]=-f A^{\dag}_{\alpha},\quad [\Delta_{\alpha\alpha},N_{\alpha}]\cong
0,\quad \ \ \Delta_{\alpha\alpha}=\Delta_{\alpha\alpha}^{\dag}.
\end{eqnarray*}
\end{proposition}
\noindent
This statement is proven straightforwardly.
\begin{proposition}\label{prop2}
For each $n\geq 0$ we have the equalities:
\begin{eqnarray}
\left[(A^{\dag}_{\alpha}A_{\alpha})^n,A^{\dag}_{\alpha}\right]=A^{\dag}_{\alpha}\left[(A^{\dag}_{\alpha}A_{\alpha}+
\varepsilon_{\alpha})^n-(A^{\dag}_{\alpha}A_{\alpha})^n\right],\label{e1}\\
\left[\varepsilon_{\alpha}^n,A^{\dag}_{\alpha}\right]=A^{\dag}_{\alpha}[(-f+\varepsilon_{\alpha})^n-\varepsilon_{\alpha}^n]. \label{e2}
\end{eqnarray}
\end{proposition}
\noindent
Using the propositions~\ref{prop1} and~\ref{prop2} and the exact commuting of $A^{\dag}_{\alpha}A_{\alpha}$ and $\varepsilon_{\alpha}$ we come to
the following
\begin{proposition}\label{prop3}
For $N_{\alpha}$ defined as $N_{\alpha}=\chi(A^{\dag}_{\alpha}A_{\alpha},\varepsilon)$, and $n\geq 0$, there is the following equality for the $n$-fold commutator (\ref{5_22}):
\begin{eqnarray*}
L_n=(A^{\dag}_{\alpha})^n\chi(A^{\dag}_{\alpha}A_{\alpha}+n\varepsilon_{\alpha}-\sigma_n
f, \varepsilon_{\alpha} - n f) -
\sum\limits_{k=0}^{n-1}C_n^k (A^{\dag}_{\alpha})^{n-k}L_k,\\
\qquad\qquad\sigma_n=\frac{n(n-1)}{2}.
\end{eqnarray*}
\end{proposition}
\noindent The proofs of propositions \ref{prop2} and \ref{prop3} are given in Appendices A and B.

Then conditions (\ref{q3}) turn into equalities
\begin{equation*}
\left\{ \eqalign{ A^{\dag}_{\alpha}\chi(A^{\dag}_{\alpha}A_{\alpha}+
\varepsilon_{\alpha},\varepsilon_{\alpha} - f)|O\rangle=
A^{\dag}_{\alpha}|O\rangle,\cr
(A^{\dag}_{\alpha})^n\chi(A^{\dag}_{\alpha}A+
n\varepsilon_{\alpha}-\sigma_n f, \varepsilon_{\alpha} - n
f)|O\rangle =\cr \qquad\qquad =C_n^1
(A^{\dag}_{\alpha})^{n-1}L_1|O\rangle =
 n(A^{\dag}_{\alpha})^{n}|O\rangle,\quad n>1.
}
\right.
\end{equation*}
To satisfy these,  it is necessary that
\begin{equation}\label{def_req}
\chi(n-\sigma_n f,1-n f) = n,\quad n\geq 1.
\end{equation}
So, the condition (\ref{def_req}) guarantees the  validity of commutation relations (\ref{n_commut}), and therefore the consistency of the whole representation of quasibosons by deformed bosons. As one can see, the both definitions $N_{\alpha}\mathop{=}\limits^{def}
\phi^{-1}(A^{\dag}_{\alpha}A_{\alpha})$ and $N_{\alpha}\mathop{=}\limits^{def}\phi^{-1}(A_{\alpha}A^{\dag}_{\alpha})-1$ satisfy (\ref{def_req}). Also, there are other definitions like $N_{\alpha}\mathop{=}\limits^{def} (1-p) \phi^{-1}(A^{\dag}_{\alpha}A_{\alpha}) + p (\phi^{-1}(A_{\alpha}A^{\dag}_{\alpha})-1)$, $0<p<1$, which satisfy (\ref{def_req}) and lead, as can be checked, to the self-consistent representation of quasibosons.

\subsection{General solution for matrices
$\Phi_{\alpha}$}\label{sec5}

In this subsection we describe how to find admissible $d_a\times d_b$ matrices $\Phi_{\alpha}$. These should satisfy the system
\begin{equation} \label{system2}
\left\{
\eqalign{
\Tr(\Phi_{\alpha}\Phi_{\beta}^{\dag})=\delta_{\alpha\beta},\cr
\Phi_{\alpha}\Phi^{\dag}_{\alpha}\Phi_{\alpha}=\frac{f}{2}\Phi_{\alpha},\cr
\Phi_{\beta}\Phi^{\dag}_{\alpha}\Phi_{\gamma}+\Phi_{\gamma}\Phi^{\dag}_{\alpha}\Phi_{\beta}=0.
}
\right.
\end{equation}
{\it Consider first the case  $f\ne 0$}. If the matrix $\Phi_\alpha$ is nondegenerate (that means  $d_a=d_b$ and $\det \Phi_{\alpha}\neq 0$) for some $\alpha$, the second relation yields $\Phi_{\alpha}\Phi^{\dag}_{\alpha}=\frac{f}{2}\mathds{1}$. From the third relation at $\gamma=\alpha$ we  obtain: $\ \Phi_{\beta}=0,\quad\forall\beta\neq\alpha. $ Then it follows that only one value of $\alpha$ is possible for which $\det \Phi_{\alpha}\neq 0$. In that case $\Phi_{\alpha}$ is an arbitrary unitary matrix. All the rest $\Phi_{\beta}=0,\ \beta\neq\alpha$. That gives the partial nondegenerate solution of the system. Note that other solutions will be degenerate for all~$\alpha$.

Let us go over to the analysis of degenerate solutions. At $\gamma=\alpha$ the last equation in (\ref{system2})  reduces  to $\Phi_{\beta} \Phi^{\dag}_{\alpha} \Phi_{\alpha} + \Phi_{\alpha} \Phi^{\dag}_{\alpha}\Phi_{\beta}=0$; multiplying it by $\Phi_{\alpha}^{\dag}$ and utilizing the second relation (note that $f$ is real) we infer
\begin{equation}\label{aa4}
K\Phi_{\beta}\Phi_{\alpha}^{\dag} \equiv
\left(\Phi_{\alpha}\Phi^{\dag}_{\alpha}+\frac{f}{2}\mathds{1}\right)\Phi_{\beta}\Phi_{\alpha}^{\dag}=0,
\quad \ K\equiv
\Phi_{\alpha}\Phi^{\dag}_{\alpha}+\frac{f}{2}\mathds{1}.
\end{equation}
From the second relation of the system (\ref{system2}) we also obtain:
\begin{equation*}
\forall x\in \Im\Phi_{\alpha}:\quad \Phi_{\alpha}\Phi^{\dag}_{\alpha}x=\frac{f}{2}x\quad\Rightarrow\quad
\dim \Im\Phi_{\alpha}\Phi^{\dag}_{\alpha}\geq \dim \Im\Phi_{\alpha}.
\end{equation*}
Taking into account the latter and the fact that $\Im\Phi_{\alpha}\Phi^{\dag}_{\alpha}\subseteq\Im\Phi_{\alpha}$ we find
\begin{equation}\label{aa3}
\Im\Phi_{\alpha}\Phi^{\dag}_{\alpha}=\Im\Phi_{\alpha}.
\end{equation}
Applying the Fredholm theorem first to $\Phi_{\alpha}$ and  then to $\Phi_{\alpha}\Phi^{\dag}_{\alpha}$ and using (\ref{aa3}) we arrive at the decompositions
\begin{eqnarray*}
\forall \alpha:\quad \mathds{C}^{d_a}=\Im\Phi_{\alpha}\oplus\Ker\Phi^{\dag}_{\alpha}=
\Im\Phi_{\alpha}\Phi^{\dag}_{\alpha}\oplus\Ker\Phi_{\alpha}\Phi^{\dag}_{\alpha},\\
\qquad\quad \mathds{C}^{d_a}=\Im\Phi_{\alpha}\oplus\Ker\,\Phi_{\alpha}\Phi^{\dag}_{\alpha}.
\end{eqnarray*}
On each of subspaces $\Im\Phi_{\alpha}$ and $\Ker\,\Phi_{\alpha}\Phi^{\dag}_{\alpha}$, which are eigenspaces for $K$, the operator $K$ is nondegenerate:
\begin{equation*}
\forall x\in \Im\Phi_{\alpha}:\quad Kx=fx,\quad{\rm and}\quad
\forall y\in\Ker\Phi_{\alpha}\Phi^{\dag}_{\alpha}:\quad
Ky=\frac{f}{2}y.
\end{equation*}
Consequently, the operator $K$ is nondegenerate on the whole $\mathds{C}^{d_a}$. Using (\ref{aa4}) we find
\begin{equation*}
\forall \alpha\neq\beta:\quad
\Phi_{\beta}\Phi_{\alpha}^{\dag}=0\quad {\rm or} \quad
\Phi_{\alpha}\Phi_{\beta}^{\dag}=0.
\end{equation*}
As a result, we arrive at the system which is equivalent to the initial one (\ref{system2}) and to the respective (for each of the
equations) implications ($\alpha\neq \beta$):
\begin{equation*}
\left\{ \eqalign{ \Tr(\Phi_{\alpha}\Phi^{\dag}_{\alpha})\!=\!1,\cr
\Phi_{\alpha}\Phi^{\dag}_{\alpha}\!\cdot\!\Phi_{\alpha}\!=\!(f/2)\!\cdot\!\Phi_{\alpha},\cr
\Phi_{\alpha}\Phi^{\dag}_{\alpha}\!\cdot\!\Phi_{\beta}\!=\!0,\cr
\Phi_{\alpha}\Phi_{\beta}^{\dag}\!=\!0. } \right. \eqalign{
\Rightarrow\ \ \dim\Im\Phi_{\alpha}\Phi^{\dag}_{\alpha}\!=\!{\rm
rank}\, \Phi_{\alpha}\!=\!2/f\equiv m,\,\ {\rm for\ all}\ \alpha,\cr
\Rightarrow\ \ \Im\Phi_{\alpha}{\rm\ -\ eigensubspace\
of}\,\Phi_{\alpha}\Phi^{\dag}_{\alpha},\cr \Rightarrow\ \
\forall\beta\neq\alpha\, \Im\Phi_{\beta}\subset
\Ker\Phi_{\alpha}\Phi^{\dag}_{\alpha}=\Ker\Phi_{\alpha}^{\dag},\cr
\Rightarrow\ \ \Im\Phi^{\dag}_{\beta}\subset \Ker\Phi_{\alpha}. }
\end{equation*}
So, the deformation parameter $f$ has a discrete range of values determined by $m$:
\begin{equation*}
f=\frac2m.
\end{equation*}
The set of the solutions depends on the relation between $\sum_{\alpha}m$ and $\min(d_a,d_b)$. If $\sum_{\alpha}m>\min(d_a,d_b)$, the set of solutions is empty. If $\sum_{\alpha}m\leq\min(d_a,d_b)$, then, according to the relations
\begin{equation*}
\mathds{C}^{d_a}=\Im\Phi_{\alpha}\oplus\Ker\Phi^{\dag}_{\alpha},\quad
\Im\Phi_{\beta}\subset
\Ker\Phi^{\dag}_{\alpha},\quad\forall\beta\neq\alpha,
\end{equation*}
the space $\mathds{C}^{d_a}$ ($\mathds{C}^{d_b}$) decomposes into the direct sum of linearly independent subspaces:
\begin{eqnarray*}
\mathds{C}^{d_a}=\biggl(\bigoplus\limits_{\alpha}
\Im\Phi_{\alpha}\biggr)\oplus R,\quad \dim
R=n-\sum\limits_{\alpha}m,\quad\Phi^{\dag}_{\alpha}R=0;\\
\mathds{C}^{d_b}=\biggl(\bigoplus\limits_{\alpha}
\Im\Phi^{\dag}_{\alpha}\biggr)\oplus \tilde{R},\quad \dim
\tilde{R}=n-\sum\limits_{\alpha}m, \quad \Phi_{\alpha}\tilde{R}=0.
\end{eqnarray*}
Let $\{e_{1\alpha},\ldots,e_{m\alpha}\}$ be the  orthonormal basis in the space $\Im\Phi_{\alpha}$, and  $U_1(d_a)$ be the corresponding transition matrix to these bases from the initial one of $\mathds{C}^{d_a}$. Likewise, let $\{f_{1\alpha},\ldots,f_{m\alpha}\}$ be the orthonormal basis in the
space $\Im\Phi^{\dag}_{\alpha}$, and $U_2(d_b)$ the corresponding transition matrix from the initial basis in $\mathds{C}^{d_b}$. In the new bases, the transition matrix $\Phi_{\alpha}$ is block-diagonal:
\begin{equation*}
U^{\dag}_1(d_a)\Phi_{\alpha}U_2(d_b)= \left(
  \begin{array}{ccc}
    0 & 0 & 0  \\
    0 & \ \tilde{\Phi}_\alpha & 0  \\
    0 & 0 & 0  \\
  \end{array}
\right)  .
\end{equation*}
The $m\times m$ matrix $\tilde{\Phi}_{\alpha}$   satisfies the equation $\tilde{\Phi}_{\alpha}\tilde{\Phi}^{\dag}_{\alpha}=\frac{f}{2} \mathds{1}_m$. Its general solution can be given through the unitary matrix: $\tilde{\Phi}_{\alpha}=\sqrt{f/2}\ U_{\alpha}(m)$. Thus the general solution of the initial system (\ref{system2}) is given in the form
\begin{equation}
\Phi_{\alpha}= U_1(d_a)
\diag\biggl\{0,\sqrt{\frac{f}{2}}U_{\alpha}(m),0\biggr\}U^{\dag}_2(d_b).\label{gen_solution}
\end{equation}
In this formula, for every matrix $\Phi_{\alpha}$,  the block $\sqrt{\frac{f}{2}}U_{\alpha}(m)$ is at its $\alpha$-th place, and does not intersect with the corresponding block of any other matrix $\Phi_{\beta}$ with $\beta\neq \alpha$. To conclude: we have got all possible quasibosonic composite operators, expressed by (\ref{anzats}) and (\ref{gen_solution}), which can be realized by the algebra of deformed oscillators.

{\it The case $f=0$ in (\ref{system2})}.  It can be shown that $\Phi_{\alpha}$ should be zero for such $f$. This is followed by applying the singular value decomposition formula for each of the matrices in the equation $\Phi_{\alpha}\Phi^{\dag}_{\alpha}\Phi_{\alpha}=0$. The fact that $\Phi_{\alpha}=0$ means, see (\ref{anzats}) and the normalization just after it, that the pure boson being a special $f=0$ case of the deformed boson with the DSF (\ref{solution1}) is unsuitable for the realization of the two-fermion composite quasiboson.

\section{Quasibosons with $q$-deformed constituent fermions}\label{sec7}

Now let us go over to the $q$-generalization of the model considered above.  Namely, we adopt nontrivial $q$-deformation for the constituents, the other assumptions being left as above. So, we start from the set of $q$-fermions, see~\cite{Viswanthan}, independent in fermionic sense:
\begin{eqnarray}
a_{\mu}a^{\dag}_{\mu'}+q^{\delta_{\mu\mu'}}a^{\dag}_{\mu'}a_{\mu}=\delta_{\mu\mu'},\qquad
b_{\nu}b^{\dag}_{\nu'}+q^{\delta_{\nu\nu'}}b^{\dag}_{\nu'}b_{\nu}=\delta_{\nu\nu'},\label{q-commut}\\
a_{\mu}a_{\mu'}+a_{\mu'}a_{\mu}=0,\ \ \mu\neq\mu',\quad
b_{\nu}b_{\nu'}+b_{\nu'}b_{\nu}=0,\ \ \nu\neq\nu'.\label{q-nez}
\end{eqnarray}
The commutation relations~(\ref{q-commut}) within one mode i.e.  for $\mu=\mu'$ and $\nu=\nu'$ completely determine the set of admissible values of the parameter $q$ and the (absence or presence, and the order of) nilpotency of the operators $a^{\dag}_{\mu}$ and $b^{\dag}_{\nu}$ depending on $q$. More precisely this is reflected in the following statement.
\begin{lemma}\label{lemma_1}
For the positivity of the norm of $q$-fermion states it is necessary to put $q\in \mathds{R}$ and $q\le 1$. If $q=1$ then $a^\dag_{\mu}$ and $b^\dag_{\nu}$ are nilpotent of second order; otherwise, if $q<1$ the operators $a^\dag_{\mu}$ and $b^\dag_{\nu}$ are not nilpotent of any order:
\begin{eqnarray}
q=1 \quad\Rightarrow\quad (a^\dag_{\mu})^2 = 0, \quad (b^\dag_{\nu})^2 = 0;\cr
q<1 \quad\Rightarrow\quad (a^\dag_{\mu})^k\neq 0,\quad (b^\dag_{\nu})^k\neq 0, \quad k\ge2.\label{q4}
\end{eqnarray}
\end{lemma}
\begin{proof}
The Lemma follows from the expression for the norm of the vector $x\!=\!(a^{\dag}_{\mu})^m |0\rangle$:
\begin{eqnarray*}
\fl ||x||^2=\langle0| a_{\mu}^k (a_{\mu}^{\dag})^k|0\rangle = \langle0|
 a_{\mu}^{k-1}[n^a_\mu+1]_{-q}(a_{\mu}^{\dag})^{k-1}|0\rangle =
  \langle0| a_{\mu}^{k-1}(a_{\mu}^{\dag})^{k-1}[n^a_\mu+k]_{-q}|0\rangle =\\
= [k]_{-q} \langle0| a_{\mu}^{k-1}(a_{\mu}^{\dag})^{k-1}|0\rangle  =
\ldots = [k]_{-q}[k-1]_{-q}\cdot...\cdot[1]_{-q},
\end{eqnarray*}
where the notation $[n]_{-q}\equiv \bigl((-q)^n-1\bigr)/\bigl((-q)-1\bigr)$ is nothing but  the deformation structure function for the $q$-fermions; $n^a_\mu$ is the number operator for $q$-fermions of $a$ type.  The same considerations apply to the operators $b^\dag_{\nu}$. This ends the proof.
\end{proof}
The $q=1$ case (i.e., usual fermions with well-known nilpotency of their creation/anni\-hilation operators) was completely analyzed in the preceding section (and also in~\cite{GKM-1}). Here we restrict ourselves to the case of $q<1$. Hence (\ref{q4}) holds for any~$k$.

The composite quasibosons' creation and annihilation operators are defined as
\begin{equation*}
A^{\dag}_{\alpha}=\sum\limits_{\mu\nu}\Phi^{\mu\nu}_{\alpha}a^{\dag}_{\mu}b^{\dag}_{\nu},
\quad
A_{\alpha}=\sum\limits_{\mu\nu}\overline{\Phi}^{\mu\nu}_{\alpha}b_{\nu}a_{\mu},
\end{equation*}
that is, like in (\ref{anzats}).  The requirements of the self-consistency of the realization (by deformed bosons) remain intact, see (\ref{system1_0}) and~(\ref{system1}):
\begin{eqnarray}
\label{treb1} A^{\dag}_{\alpha}A_{\alpha}\cong \phi(N_{\alpha}),\quad A_{\alpha}A^{\dag}_{\alpha}\cong \phi(N_{\alpha}+1),\\
\label{treb2} [A^{\dag}_{\alpha},A^{\dag}_{\beta}]\cong 0\ \Leftrightarrow\ [A_{\alpha},A_{\beta}]\cong 0,\quad [A_{\alpha},A^{\dag}_{\beta}]\cong0,\ \alpha\neq\beta,\\
\label{treb3} [N_{\alpha},A^{\dag}_{\alpha}]\cong A^{\dag}_{\alpha},\quad [N_{\alpha},A_{\alpha}]\cong -A_{\alpha}.
\end{eqnarray}
In this case the requirement of independence $[A^{\dag}_{\alpha},A^{\dag}_{\beta}]\cong 0$, as one can easily check, leads to the following condition on matrices $\Phi_{\alpha}$:
\begin{equation}\label{nez1}
\Phi^{\mu\nu}_{\alpha}\Phi^{\mu\nu'}_{\beta} = \Phi^{\mu\nu'}_{\alpha}\Phi^{\mu\nu}_{\beta},\quad \Phi^{\mu\nu}_{\alpha}\Phi^{\mu'\nu}_{\beta} = \Phi^{\mu'\nu}_{\alpha}\Phi^{\mu\nu}_{\beta}.
\end{equation}
The second relation in (\ref{treb1}) implies that there should be
\begin{equation}
A_{\alpha}(A^{\dag}_{\alpha})^n |O\rangle =
\phi(N_{\alpha}+1)(A^{\dag}_{\alpha})^{n-1}|O\rangle,\qquad \
n=1,2,3...\,.\label{system3}
\end{equation}
Using (\ref{treb3}) we obtain:
\[
\phi(N_{\alpha}+1)(A^{\dag}_{\alpha})^{n-1}|O\rangle = (A^{\dag}_{\alpha})^{n-1}\phi(N_{\alpha}+n)|O\rangle.
\]
As a result, we arrive at
\begin{equation}
A_{\alpha}(A^{\dag}_{\alpha})^n |O\rangle = \phi(n) (A^{\dag}_{\alpha})^{n-1}|O\rangle,\quad n=1,2,3...\,
.\label{system3-2}
\end{equation}
It can be checked by induction that
\begin{eqnarray*}
A_{\alpha}(A^{\dag}_{\alpha})^n =
(-1)^{\left[\frac{n-1}{2}\right]}\overline {\Phi_{\alpha}^{\mu\nu}}
\prod_{j=1}^n \Phi_{\alpha}^{\mu_1\nu_1}\cdot\\
\cdot\biggl[\sum\limits_{i=1}^n (-1)^{i-1}
\delta_{\mu\mu_i}q^{\sum\limits_{s=1}^{i-1}\delta_{\mu\mu_s}}
\prod_{\scriptstyle r=1\atop\scriptstyle r\neq i}^n a^{\dag}_{\mu_r}
+ (-1)^n q^{\sum\limits_{s=1}^n \delta_{\mu\mu_s}} \prod_{r=1}^n
a^{\dag}_{\mu_r}\cdot a_{\mu}\biggr]\cdot\\
\cdot\biggl[\sum\limits_{k=1}^n (-1)^{k-1}
\delta_{\nu\nu_k}q^{\sum\limits_{s=1}^{k-1} \delta_{\nu\nu_s}}
\prod_{\scriptstyle r=1\atop\scriptstyle r\neq k}^n b^{\dag}_{\nu_r}
+ (-1)^n q^{\sum\limits_{s=1}^n \delta_{\nu\nu_s}} \prod_{r=1}^n
b^{\dag}_{\nu_r}\cdot b_{\nu}\biggr].
\end{eqnarray*}
Then, using equation (\ref{system3-2}) we arrive at
\begin{eqnarray}
\phi(n)\prod_{l=1}^{n-1} \Phi_{\alpha}^{\mu_l\nu_l}
a^{\dag}_{\mu_l}b^{\dag}_{\nu_l}|O\rangle =
(-1)^{\left[\frac{n-1}{2}\right]}\overline{\Phi_{\alpha}^{\mu\nu}}
\prod_{l=1}^n \Phi_{\alpha}^{\mu_l\nu_l}\cdot\nonumber\\
\cdot\biggl[\sum\limits_{i=1}^n (-1)^{i-1}
\delta_{\mu\mu_i}q^{\sum\limits_{s=1}^ {i-1} \delta_{\mu\mu_s}}
\prod_{\scriptstyle r=1\atop\scriptstyle r\neq i}^n a^{\dag}_{\mu_r}
+ (-1)^n q^{\sum\limits_{s=1}^n \delta_{\mu\mu_s}} \prod_{r=1}^n a^{\dag}_{\mu_r}\cdot a_{\mu}\biggr]\cdot\nonumber\\
\cdot\biggl[\sum\limits_{k=1}^n (-1)^{k-1}
\delta_{\nu\nu_k}q^{\sum\limits_{s=1}^{k-1} \delta_{\nu\nu_s}}
\prod_{\scriptstyle r=1\atop\scriptstyle r\neq k}^n b^{\dag}_{\nu_r}
+ (-1)^n q^{\sum\limits_{s=1}^n \delta_{\nu\nu_s}} \prod_{r=1}^n
b^{\dag}_{\nu_r}\cdot b_{\nu}\biggr]|O\rangle.\label{vspm}
\end{eqnarray}
Note that if (\ref{vspm}) holds on the vacuum, the following equality holds on any state:
\begin{eqnarray}
(-1)^{\left[\frac{n-1}{2}\right]}\overline {\Phi_{\alpha}^{\mu\nu}}
\prod_{l=1}^n \Phi_{\alpha}^{\mu_l\nu_l}\cdot
\biggl[\sum\limits_{i=1}^n (-1)^{i-1}
\delta_{\mu\mu_i}q^{\sum\limits_{s=1}^{i-1} \delta_{\mu\mu_s}}
\prod_{\scriptstyle r=1\atop\scriptstyle r\neq i}^n a^{\dag}_{\mu_r}
\biggr]\cdot\nonumber\\
\cdot\biggl[\sum\limits_{k=1}^n (-1)^{k-1}
\delta_{\nu\nu_k}q^{\sum\limits_{s=1}^{k-1} \delta_{\nu\nu_s}}
\prod_{\scriptstyle r=1\atop\scriptstyle r\neq k}^n
b^{\dag}_{\nu_r}\biggr]=\phi(n)\prod_{l=1}^{n-1}
\Phi_{\alpha}^{\mu_l\nu_l}a^{\dag}_{\mu_l}b^{\dag}_{\nu_l}\,.\label{vspm2}
\end{eqnarray}
As a recursive step, let us consider the following relation valid for $n+1$:
\begin{eqnarray*}
\fl A_{\alpha}(A^{\dag}_{\alpha})^{n+1} =
(-1)^{\left[\frac{n}{2}\right]}
\overline{\Phi_{\alpha}^{\mu\nu}}\prod_{l=1}^n
\Phi_{\alpha}^{\mu_l\nu_l} \cdot\biggl[\sum\limits_{i=1}^{n+1}
(-1)^{i-1} \delta_{\mu\mu_i}q^{\sum\limits_{s=1}^{i-1}
\delta_{\mu\mu_s}} \prod_{\scriptstyle r=1\atop\scriptstyle r\neq
i}^{n+1} a^{\dag}_{\mu_r}\biggr]\cdot\\
\cdot\biggl[\sum\limits_{k=1}^{n+1} (-1\!)^{k-1}
\delta_{\nu\nu_k}q^{\sum\limits_{s\!=\!1}^{k\!-\!1} \!\delta_{\nu\nu_s}}
\!\!\prod_{\scriptstyle r=1\atop\scriptstyle r\neq k}\!
b^{\dag}_{\nu_r}\biggr] \Phi_{\alpha}^{\mu_{n\!+\!1}\nu_{n\!+\!1}} \!+\!
(-1\!)^{\left[\frac{n-1}{2}\right]} \overline{\Phi_{\alpha}^{\mu\nu}}
\!\prod_{l=1}^n\! \Phi_{\alpha}^{\mu_l\nu_l}\!\cdot \\
\biggl[\biggl(\sum\limits_{i=1}^{n+1} (-1)^{i-1}
\delta_{\mu\mu_i}q^{\sum\limits_{s=1}^{i-1} \delta_{\mu\mu_s}}
\prod_{\scriptstyle r=1\atop\scriptstyle r\neq i}
a^{\dag}_{\mu_r}\biggr)\cdot (-1)^n q^{\sum\limits_{s=1}^{n+1}
\delta_{\nu\nu_s}} \prod_{r=1}^n b^{\dag}_{\nu_r}\cdot b_{\nu} + \\
+ (-1)^n q^{\sum\limits_{s=1}^{n+1} \delta_{\mu\mu_s}} \prod_{r=1}^n
a^{\dag}_{\mu_r}\cdot a_{\mu} \biggl(\sum\limits_{k=1}^{n+1}
(-1)^{k-1}
\delta_{\nu\nu_k}q^{\sum\limits_{s=1}^{k-1}\delta_{\nu\nu_s}}
\prod_{\scriptstyle r=1\atop\scriptstyle r\neq k}^{n}
b^{\dag}_{\nu_r}\biggr)+\\
+ q^{\sum\limits_{s=1}^{n+1} \delta_{\mu\mu_s}+\delta_{\nu\nu_s}}
\prod_{r=1}^n a^{\dag}_{\mu_r} \cdot a_{\mu} \prod_{r=1}^n
b^{\dag}_{\nu_r}\cdot b_{\nu}\biggr]
\Phi_{\alpha}^{\mu_{n+1}\nu_{n+1}}a^{\dag}_{\mu_{n+1}}b^{\dag}_{\nu_{n+1}}
\mathop{=}\limits^{(\ref{vspm2})}\\
\mathop{=}\limits^{(\ref{vspm2})}
\phi(n)\!\prod_{l=1}^n\! \Phi_{\alpha}^{\mu_l\nu_l}
a^{\dag}_{\mu_l}b^{\dag}_{\nu_l} \!+\!
(-1\!)^{\left[\frac{n}{2}\right]}\overline{\Phi_{\alpha}^{\mu\nu}}
\!\prod_{l=1}^{n+1}\! \Phi_{\alpha}^{\mu_l\nu_l}\cdot\\
\Biggl[(-1\!)^n\!\biggl(\!\sum\limits_{i=1}^n (-1\!)^{i\!-\!1} \delta_{\mu\mu_i}
q^{\sum\limits_{s\!=\!1}^{i\!-\!1} \!\delta_{\mu\mu_s}} \!\!\prod_{\scriptstyle
r=1\atop\scriptstyle r\neq i}\! a^{\dag}_{\mu_r}\!\biggr)
\Bigl(\!\delta_{\nu\nu_{n\!+\!1}}q^{\sum\limits_{s\!=\!1}^n\!
\delta_{\nu\nu_s}} \!\prod_{r=1}^n\! b^{\dag}_{\nu_r} \!\!-\! q^{\sum\limits_{s\!=\!1}^{n\!+\!1} \!\delta_{\nu\nu_s}}
\!\prod_{r=1}^{n+1}\! b^{\dag}_{\nu_r}\!\!\cdot\! b_{\nu}\!\Bigr) \!+\\
+ (-1\!)^n\!\Bigl(\!\delta_{\mu\mu_{n\!+\!1}}q^{\sum\limits_{s\!=\!1}^n \!\delta_{\mu\mu_s}}
\!\!\!\prod_{r=1}^n a^{\dag}_{\mu_r}
\!\!-\!  q^{\sum\limits_{s\!=\!1}^{n\!+\!1} \!\delta_{\mu\mu_s}}
\!\!\!\prod_{r=1}^{n+1}\! a^{\dag}_{\mu_r}\!\!\cdot\! a_{\mu}\!\Bigr)
\biggl(\!\sum\limits_{k=1}^n (-1\!)^{k\!-\!1}
\delta_{\nu\nu_k}q^{\sum\limits_{s\!=\!1}^{k\!-\!1} \!\delta_{\nu\nu_s}}
\!\!\!\prod_{\scriptstyle r=1\atop\scriptstyle r\neq k}^{n+1}\! b^{\dag}_{\nu_r}\!\biggr)\!+\\
+ \Bigl(\!\delta_{\mu\mu_{n\!+\!1}}q^{\sum\limits_{s\!=\!1}^n\!
\delta_{\mu\mu_s}} \!\!\!\prod_{r=1}^n\! a^{\dag}_{\mu_r} \!\!-\!
q^{\sum\limits_{s\!=\!1}^{n\!+\!1} \!\delta_{\mu\mu_s}} \!\!\!\prod_{r=1}^{n+1}\!
a^{\dag}_{\mu_r}\!\!\cdot\! a_{\mu}\!\Bigr)
\Bigl(\!\delta_{\nu\nu_{n\!+\!1}}q^{\sum\limits_{s\!=\!1}^n
\!\delta_{\nu\nu_s}} \!\!\!\prod_{r=1}^n b^{\dag}_{\nu_r}\!\! -\!
q^{\sum\limits_{s\!=\!1}^{n\!+\!1} \!\delta_{\nu\nu_s}} \!\!\!\prod_{r=1}^{n+1}\!
b^{\dag}_{\nu_r}\!\!\cdot\! b_{\nu}\!\Bigr)\Biggr]
\end{eqnarray*}
where at the last stage we have used (\ref{vspm2}). Substituting the last expression for $A_{\alpha}(A^{\dag}_{\alpha})^{n+1}$ into (\ref{system3-2}) rewritten for $n\!\rightarrow\! n\!+\!1$ we deduce the following relation that involves the linear combination:
\vspace{-1.5mm}
\begin{equation}
\sum_{\mu_1...\mu_n,\nu_1...\nu_n} B^{\mu_1...\mu_n,\nu_1...\nu_n}(\Phi_{\alpha},q) \cdot
e_{\mu_1...\mu_n,\nu_1...\nu_n} = 0,\label{rav3}
\end{equation}
\vspace{-0.5mm}
where the coefficients are
\vspace{-0.5mm}
\begin{eqnarray*}
\fl B^{\mu_1...\mu_n,\nu_1...\nu_n}(\Phi_{\alpha},q) = -\sum\limits_{i=1}^{n} q^{\sum\limits_{s=1}^{i-1}
(\delta_{\mu\mu_s}+\delta_{\nu\nu_s})}
\Phi_{\alpha}^{\mu_n\nu}\overline{\Phi_{\alpha}^{\mu\nu}}\Phi_{\alpha}^{\mu\nu_n} \prod_{l=1}^{n-1}
\Phi_{\alpha}^{\mu_l\nu_l}\cdot\\
\Bigl((-1)^{\sum\limits_{r=i}^{n-1} \delta_{\nu_r\nu_{r+1}}}
q^{\sum\limits_{s=i}^n \delta_{\nu\nu_s}}
+ (-1)^{\sum\limits_{r=i}^{n-1} \delta_{\mu_r\mu_{r+1}}}
q^{\sum\limits_{s=i}^n \delta_{\mu\mu_s}}\Bigr) +\\
+q^{\sum\limits_{s=1}^n
(\delta_{\mu\mu_s}+\delta_{\nu\nu_s})}\overline{\Phi_{\alpha}^{\mu\nu}}
\Phi_{\alpha}^{\mu\nu} \prod_{l=1}^n \Phi_{\alpha}^{\mu_l\nu_l} -
\left[\phi(n+1)- \phi(n)\right]\prod_{l=1}^n
\Phi_{\alpha}^{\mu_l\nu_l}
\end{eqnarray*}
and the basis elements are
\begin{equation*}
e_{\mu_1...\mu_n,\nu_1...\nu_n} = a^{\dag}_{\mu_1}b^{\dag}_{\nu_1}... a^{\dag}_{\mu_n}b^{\dag}_{\nu_n}|O\rangle.
\end{equation*}
These basis elements are independent for differing sets of indices $\mu_1...\mu_n$ and $\nu_1...\nu_n$ regardless of any permutations within each set. So let us extract in (\ref{rav3}) the terms with $\mu_1=\ldots=\mu_n$ and $\nu_1=\ldots=\nu_n$; using their linear independence from the others, we infer $B^{\mu_1...\mu_1,\nu_1...\nu_1}(\Phi_{\alpha},q)=0$, that can be rewritten in the following form:
\begin{eqnarray*}
\fl \sum\limits_{i=1}^{n} (\!-1\!)^{n\!+i\!-\!1}\!\Bigl(\!2\!+\!(\delta_{\mu\mu_1}\!\!+\!
\delta_{\nu\nu_1})(q^n\!\!-\!q^{i\!-\!1}\!\!-\!2)\!+\!2\delta_{\mu\mu_1}\!\delta_{\nu\nu_1}(q^{n}\!\!-\!1)(q^{i\!-\!1}\!\!-\!1)\!\Bigr)\!
\Phi_{\alpha}^{\mu_1\nu}\overline{\Phi_{\alpha}^{\mu\nu}}\!
\Phi_{\alpha}^{\mu\nu_1}\!(\!\Phi_{\alpha}^{\mu_1\nu_1}\!)^{n\!-\!1} \!\!+\! \\
+\Bigl(1+(\delta_{\mu\mu_1}+\delta_{\nu\nu_1})(q^n-1)+
\delta_{\mu\mu_1}\delta_{\nu\nu_1}(q^n-1)^2\Bigr)
\overline{\Phi_{\alpha}^{\mu\nu}}\Phi_{\alpha}^{\mu\nu}(\Phi_{\alpha}^{\mu_1\nu_1})^{n}=\\
\qquad\qquad=[\phi(n+1)-\phi(n)](\Phi_{\alpha}^{\mu_1\nu_1})^{n} .
\end{eqnarray*}
Performing the summation over $i$, $\mu$, $\nu$ on the left-hand side, we find
\begin{eqnarray}
\fl ((-1)^n\!-\!1)\left(\Phi_{\alpha}\Phi_{\alpha}^{\dag}\Phi_{\alpha}\right)^{\mu_1\nu_1}\!\!(\Phi_{\alpha}^{\mu_1\nu_1})^{n-1}\!
+\Bigl(\frac12 (-q)^n + \frac{q\!-\!1}{2(q\!+\!1)}q^n - \frac{q}{q\!+\!1}(-1)^n\Bigr)\!\bigl[(\Phi_{\alpha}^{\dag}\Phi_{\alpha})^{\nu_1\nu_1} \!\!+\nonumber\\
+ (\Phi_{\alpha}\Phi_{\alpha}^{\dag})^{\mu_1\mu_1}\bigr](\Phi_{\alpha}^{\mu_1\nu_1})^{n}
+\frac{q-1}{q+1}(q^n-1)\left(q^n-(-1)^n\right)|\Phi_{\alpha}^{\mu_1\nu_1}|^2(\Phi_{\alpha}^{\mu_1\nu_1})^{n}=\nonumber\\
\qquad\qquad=[\phi(n+1)-\phi(n)-1](\Phi_{\alpha}^{\mu_1\nu_1})^{n}.
\label{usl5}
\end{eqnarray}
For all the indices $(\mu_1,\nu_1)$ for  which $\Phi_{\alpha}^{\mu_1\nu_1}\neq 0$, the last equation can be divided by $(\Phi_{\alpha}^{\mu_1\nu_1})^{n}$. Summing (\ref{usl5}) over $n$ from $n=1$ to $n=s$ and then replacing in the resulting equality $s\rightarrow n-1$ we obtain:
\begin{eqnarray*}
\fl \Bigl(\frac{1-(-1)^{n}}{2}-n\Bigr)
\frac{(\Phi_{\alpha}\Phi^{\dag}_{\alpha}
\Phi_{\alpha})^{\mu_1\nu_1}}{\Phi_{\alpha}^{\mu_1\nu_1}}+
\Bigl([n]_{-q}-\frac{1-(-1)^n}{2}\Bigr)^2 \cdot|\Phi_{\alpha}^{\mu_1\nu_1}|^2+\\
+\frac{1-(-1)^n}{2}\bigl([n]_{-q}-1\bigr)
\left[(\Phi_{\alpha}^{\dag}\Phi_{\alpha})^{\nu_1\nu_1} +
(\Phi_{\alpha}\Phi_{\alpha}^{\dag})^{\mu_1\mu_1}\right] =
\phi(n)-n,\ n\geq 2.
\end{eqnarray*}
Note that the functions $\left(\frac{1-(-1)^{n}}{2}-n\right)$, $\left([n]_{-q}-\frac{1-(-1)^n}{2}\right)^2$ and  $\frac{1-(-1)^n}{2}\bigl([n]_{-q}-1\bigr)$ as functions of $n$ are independent for the admissible values of $q$. Hence
$(\Phi_{\alpha}\Phi^{\dag}_{\alpha}\Phi_{\alpha})^{\mu_1\nu_1}/\Phi_{\alpha}^{\mu_1\nu_1}$,
$|\Phi_{\alpha}^{\mu_1\nu_1}|^2$ and
$[(\Phi_{\alpha}^{\dag}\Phi_{\alpha})^{\nu_1\nu_1} +
(\Phi_{\alpha}\Phi_{\alpha}^{\dag})^{\mu_1\mu_1}]$ do not depend on
$(\mu_1,\nu_1)$ if $\Phi_{\alpha}^{\mu_1\nu_1}\neq 0$:
\begin{eqnarray*}
{(\Phi_{\alpha}\Phi^{\dag}_{\alpha}\Phi_{\alpha})^{\mu_1\nu_1}}/{\Phi_{\alpha}^{\mu_1\nu_1}}=p_1,\\
|\Phi_{\alpha}^{\mu_1\nu_1}|^2=p_2,\\
(\Phi_{\alpha}^{\dag}\Phi_{\alpha})^{\nu_1\nu_1} +
(\Phi_{\alpha}\Phi_{\alpha}^{\dag})^{\mu_1\mu_1}=p_3,
\end{eqnarray*}
where $p_1$, $p_2$ and $p_3$ are some numerical parameters. Thus, we obtain
\begin{eqnarray}
\fl \phi(n)=n-\Bigl(n\!-\!\frac{1\!-\!(-1)^{n}}{2}\Bigr)p_1+\Bigl([n]_{-q}\!-\!\frac{1\!-\!(-1)^n}{2}\Bigr)^2p_2 + \frac{1\!-\!(-1)^n}{2}\bigl([n]_{-q}\!-\!1\bigr)p_3.\label{phi2}
\end{eqnarray}
\par
Let us now consider the terms in equation (\ref{rav3}) with $n$ equated indices $\mu_1=\ldots=\mu_n$ and with $n-1$ equated indices in the set ($\nu_1,\ldots,\nu_n$), the remaining one being different. Denote the $n-1$ equal indices by $\nu_1$, and the differing one (suppose it occupies the $k$th position) by $\nu_2$. Due to the linear independence of the mentioned terms from the others we obtain the equation
\begin{equation}
\fl \sum_{k=1}^n B^{\mu_1...\mu_1,\nu_1..\nu_k..\nu_1} e_{\mu_1...\mu_1,\nu_1..\nu_k..\nu_1}|_{\nu_k\rightarrow \nu_2} = 0\ \ {\rm i.e.}\ \sum_{k=1}^n (-1)^k B^{\mu_1...\mu_1,\nu_1..\nu_k..\nu_1}|_{\nu_k\rightarrow \nu_2}=0.\label{eq1}
\end{equation}
Introducing auxiliary notations
\begin{equation*} \left\{
\eqalign{
X=\Phi_{\alpha}^{\mu_1\nu_1}\Phi_{\alpha}^{\mu_1\nu_2},\\
Y=\Phi_{\alpha}^{\mu_1\nu_1}\Phi_{\alpha}^{\mu_1\nu_2}
(\Phi_{\alpha}\Phi_{\alpha}^{\dag})^{\mu_1\mu_1},\\
Z=(\Phi_{\alpha}^{\mu_1\nu_1})^2(\Phi_{\alpha}^{\dag}\Phi_{\alpha})^{\nu_1\nu_2},
}
\right.
\end{equation*}
after performing all the summations in (\ref{eq1}) we obtain:
\begin{eqnarray*}
\eqalign{
[Xp_2]
(-1)^nq^{2n}+\\
[(-q^3\!+\!2q^2\!-\!3q\!+\!4)p_2X]
q^{2n}+\\
[((q^2\!-\!2\!-\!q)p_2+2p_3)X+(-\!2\!-\!q)Y]
nq^n+\\
[((-\!3q^3\!+\!17q^2\!+\!q^4\!-\!26\!-\!5q)p_2\!+\!(-\!4q\!+\!2q^2\!+\!2)p_3)X\!+\!(6\!+\!5q\!-\!q^3\!-\!2q^2)Y\!+\\
\qquad\qquad\qquad+(4q\!-\!10q^2\!+\!6)Z]q^n+\\
[((-\!q^3\!+\!q\!+\!2q^2\!-\!2)p_2+(2\!-\!2q)p_3)X+(q^2\!+\!3q\!-\!2)Y+(-\!2\!-\!2q)Z]
(-q)^n+\\
[((q^2\!+\!3\!-\!4q)p_2+(-\!4q^2\!+\!2q\!+\!2)p_3)X+(4q^2\!-\!q\!-\!5)Y+(3q^2\!+\!1)Z]
(-1)^n+\\
[(2p_1+(-\!3q\!+\!5)p_2-2p_3)X+Y+(3q\!-\!3)Z]
n+\\
[((8\!-\!8q^2)p_1+(23\!-\!3q\!-\!19q^2\!+\!7q^3)p_2+(8q\!-\!4q^3\!+\!2q^2\!-\!6)p_3)X+\\
\qquad\qquad\qquad+(4q^3\!-\!5\!-\!12q\!+\!5q^2)Y+(-\!3q^3\!-\!11\!+\!3q\!+\!11q^2)Z]=0.
}
\end{eqnarray*}
Extracting the coefficients of this system  at the linearly independent functions $(-1)^nq^{2n}$, $q^{2n}$, $nq^n$, $q^n$, $(-q)^n, (-1)^n, n, 1$ (considered as the elements of the vector space of functions of $n$), we arrive at  the following system:
\[
\left\{
\eqalign{
Xp_2=0,\cr
[-q^3+2q^2-3q+4]p_2X=0,\cr
[(q^2\!-\!2\!-\!q)p_2+2p_3]X+[-2-q]Y=0,\cr
[(-\!3q^3\!+\!17q^2\!+\!q^4\!-\!26\!-\!5q)p_2\!+\!(-\!4q\!+\!2q^2\!+\!2)p_3]X\!+\![6\!+\!5q\!-\!q^3\!-\!2q^2]Y\!+\cr
\qquad\qquad\qquad+[4q\!-\!10q^2\!+\!6]Z=0,\cr
[(-\!q^3\!+\!q\!+\!2q^2\!-\!2)p_2+(2\!-\!2q)p_3]X+[q^2\!+\!3q\!-\!2]Y+[-\!2\!-\!2q]Z=0,\cr
[(q^2\!+\!3\!-\!4q)p_2+(-\!4q^2\!+\!2q\!+\!2)p_3]X+[4q^2\!-\!q\!-\!5]Y+[3q^2\!+\!1]Z=0,\cr
[2p_1+(-\!3q\!+\!5)p_2\!-\!2p_3]X+Y+[3q\!-\!3]Z=0,\cr
[(8\!-\!8q^2)p_1\!+\!(23\!-\!3q\!-\!19q^2\!+\!7q^3)p_2\!+\!(8q\!-\!4q^3\!+\!2q^2\!-\!6)p_3]X+\cr
\,\qquad\qquad\qquad+[4q^3\!-\!5\!-\!12q\!+\!5q^2]Y+[-\!3q^3\!-\!11\!+\!3q\!+\!11q^2]Z=0.
}
\right.
\]
The solution of this system is ($q\neq 1$):
\begin{equation*}
\left\{
\eqalign{
X=\Phi_{\alpha}^{\mu_1\nu_1}\Phi_{\alpha}^{\mu_1\nu_2}=0,\cr
Y=\Phi_{\alpha}^{\mu_1\nu_1}\Phi_{\alpha}^{\mu_1\nu_2}(\Phi_{\alpha}\Phi_{\alpha}^{\dag})^{\mu_1\mu_1}=0,\cr
Z=(\Phi_{\alpha}^{\mu_1\nu_1})^2(\Phi_{\alpha}^{\dag}\Phi_{\alpha})^{\nu_1\nu_2}=0.
}
\right.
\end{equation*}
This set of conditions is equivalent to
\begin{equation}\label{aa1}
\Phi_{\alpha}^{\mu_1\nu_1}\Phi_{\alpha}^{\mu_1\nu_2}=0,
\end{equation}
which means that the matrix $\Phi_{\alpha}$ cannot contain two nonzero elements in any one row.

In a similar way we can derive the condition
\begin{equation}\label{aa2}
\Phi_{\alpha}^{\mu_1\nu_1}\Phi_{\alpha}^{\mu_2\nu_1}=0,
\end{equation}
implying that the matrix $\Phi_{\alpha}$ cannot contain two nonzero elements in any one column.

Next, the same analysis as in the previous two paragraphs is performed for those terms in~(\ref{rav3}), for which:
in the set ($\mu_1,\ldots,\mu_n$) there is only one index (denoted by $\mu_2$) different from the other $(n-1)$ equal ones (all denoted by $\mu_1$), and likewise for $\nu$-indices -- in the set ($\nu_1,\ldots,\nu_n$) there is only one index (denoted by $\nu_2$) different from the other, equal ones (all denoted by $\nu_1$). As a result, we derive
\begin{equation}
\Phi_{\alpha}^{\mu_1\nu_1}\Phi_{\alpha}^{\mu_2\nu_2}=0.  \label{aa3}
\end{equation}
That is, the matrix $\Phi_{\alpha}$ cannot have two nonzero elements in differing rows and columns. And, using the previous conditions (\ref{aa1}) and (\ref{aa2}) we obtain that the matrix $\Phi_{\alpha}$ cannot contain two nonzero elements. As a consequence, we obtain the following values for the parameters $p_1,p_2,p_3$:
\[
p_1=p_2=1,\quad p_3=2.
\]
Then the following expression for the deformation structure function results from (\ref{phi2}):
\begin{equation}\label{str_f2}
\phi(n) = \left([n]_{-q}\right)^2.
\end{equation}
The mode-independence conditions contained in (\ref{nez1}) and equalities (\ref{aa1}), (\ref{aa2}) and (\ref{aa3}) enable us to determine the solution for $\Phi_{\alpha}$: the only nonzero elements in matrices $\Phi_{\alpha}$ and $\Phi_{\beta}$ are situated at the intersection of different rows and different columns:
\begin{equation}
\Phi_{\alpha}^{\mu\nu}=\Phi_{\alpha}^{\mu_0(\alpha)\nu_0(\alpha)}\delta_{\mu\mu_0(\alpha)}\delta_{\nu\nu_0(\alpha)},\qquad |\Phi_{\alpha}^{\mu_0(\alpha)\nu_0(\alpha)}|=1.\label{matrPhi2}
\end{equation}

For the illustrative purpose, a more detailed treatment of two particular examples including also the omitted steps of the derivation above, is provided in~\ref{ap3}. The first example concerns the case with only one possible value of $\mu,\nu=1$ for the constituent $q$-fermions modes. The second example concerns the case of two-mode constituents, i.e. of two possible values of~$\mu,\nu=\overline{1,2}$.

It remains to satisfy the commutation relations~(\ref{treb3}) by means of the correct definition of $N_{\alpha}$.  Let $N_{\alpha}$ be defined as $N_{\alpha} \mathop{=}\limits^{def} \chi(A^{\dag}_{\alpha}A_{\alpha}, A_{\alpha}A^{\dag}_{\alpha})$, and the matrices $\Phi_{\alpha}$ are those already found in~(\ref{matrPhi2}). Taking into account the latter we have
\begin{equation*}
A_{\alpha}A^{\dag}_{\alpha}\!\cdot\! (A^{\dag}_{\alpha})^n |O\rangle = [n\!+\!1]^2_{-q} (A^{\dag}_{\alpha})^n |O\rangle,\quad A^{\dag}_{\alpha}A_{\alpha}\!\cdot\! (A^{\dag}_{\alpha})^n |O\rangle = [n]^2_{-q} (A^{\dag}_{\alpha})^n |O\rangle.
\end{equation*}
Then~(\ref{treb3}) is equivalent to
\begin{eqnarray*}
\fl \chi(A^{\dag}_{\alpha}A_{\alpha}, A_{\alpha}A^{\dag}_{\alpha})(A^{\dag}_{\alpha})^{n+1}|O\rangle - A^{\dag}_{\alpha} \chi(A^{\dag}_{\alpha}A_{\alpha}, A_{\alpha}A^{\dag}_{\alpha}) (A^{\dag}_{\alpha})^n|O\rangle = A^{\dag}_{\alpha} (A^{\dag}_{\alpha})^n|O\rangle \Leftrightarrow\\
\fl \Leftrightarrow \chi(A^{\dag}_{\alpha}A_{\alpha}, [n+2]^2_{-q})(A^{\dag}_{\alpha})^{n+1}|O\rangle - A^{\dag}_{\alpha} \chi(A^{\dag}_{\alpha}A_{\alpha}, [n+1]^2_{-q}) (A^{\dag}_{\alpha})^n|O\rangle = (A^{\dag}_{\alpha})^{n+1}|O\rangle \Leftrightarrow\\
\fl \Leftrightarrow\chi([n+1]^2_{-q}, [n+2]^2_{-q})(A^{\dag}_{\alpha})^{n+1}|O\rangle - \chi([n]^2_{-q}, [n+1]^2_{-q}) (A^{\dag}_{\alpha})^{n+1}|O\rangle = (A^{\dag}_{\alpha})^{n+1}|O\rangle \Leftrightarrow\\
\Leftrightarrow \chi([n+1]^2_{-q}, [n+2]^2_{-q}) - \chi([n]^{-q}, [n+1]^2_{-q})  = 1,\ \ n\ge0.
\end{eqnarray*}
Thus the condition $\chi\bigl([n]_{-q}^{\,2},[n+1]_{-q}^{\,2}\bigr)\bigr|_n^{n+1}\equiv \chi([n+1]^2_{-q}, [n+2]^2_{-q}) - \chi([n]^2_{-q}, [n+1]^2_{-q}) =1$, $n=0,1,...$, is necessary and sufficient for~(\ref{treb3}) to hold.
\begin{remark}\label{rem3}
Expression~(\ref{str_f2}) for the structure function is valid only when $q\ne 1$ i.e. when $a_{\mu}^\dag$, $a_{\nu}^\dag$ are not nilpotent of any order. If $q=1$, it is the DSF~(\ref{solution1}) which provides the realization. Thus, the unifying formula for the deformation structure function (of those deformed oscillators that give realization) for quasibosons composed of two $q$-fermions can be written as
\begin{equation}\label{gen_DSF}
\phi(n) = \left\{
\eqalign{
\left([n]_{-q}\right)^2=\Bigl(\frac{1-(-q)^{n}}{1+q}\Bigr)^2,\quad q<1;\cr
\Bigl(1+\frac1m\Bigr)n - \frac1m n^2,\qquad q=1,\quad m\in\mathds{N}.
}
\right.
\end{equation}
The absence of a continuous limit from (\ref{str_f2}) to (\ref{solution1}) when $q\rightarrow 1$ or in other words the discontinuity of (\ref{gen_DSF}) at the $q=1$ point is explained as follows. If $q\ne 1$ then there is an infinite number of basis elements \{$(a_1^\dag)^{k_1}...(a_{d_a}^\dag)^{k_{d_a}}(b_1^\dag)^{l_1}...(b_{d_b}^\dag)^{l_{d_b}} |O\rangle$ $\bigr|$ $k_i,l_j\ge 0$, $\sum_{i=1}^{d_a} k_i = \sum_{j=1}^{d_b} l_j = n$, $n=0,1,2,...$\} of the subspace of composite bosons' states. The latter results in an infinite number of requirements~(\ref{rav3}) thus imposing a considerable restriction on $\Phi_{\alpha}^{\mu\nu}$. On the other hand, if $q=1$, then there is only finite number, equal to $\sum_{k=1}^{\min(d_a,d_b)} C_{d_a}^kC_{d_b}^k = C_{d_a+d_b}^{\max(d_a,d_b)}$ of the basis elements: $|O\rangle$, $a^{\dag}_{\mu}b^{\dag}_{\nu}|O\rangle$,\,\,...\,\,, $a_1^\dag ...a_{\min(d_a,d_b)}^\dag b_1^\dag... b_{\min(d_a,d_b)}^\dag|O\rangle$, that leads to a finite number of requirements~(\ref{rav3}). Moreover, in this case only a few requirements among them are independent, see~(\ref{system2}).
\end{remark}

\section{Conclusions and outlook}

As shown in our preceding paper~\cite{GKM-1} and in Section~\ref{sec2} above, the problem of realization of "fermion+fermion" quasibosons by means of deformed oscillators has nontrivial solutions. In the case of pure fermions as constituents, the structure function $\phi$ of the relevant deformation is found in the form (\ref{solution1}) quadratic in the number operator $N$, with a discrete valued deformation parameter $f=2/m$. This is the only
DSF for which the realization (isomorphism) is possible. In addition, necessary and sufficient conditions on the matrices $\Phi_{\alpha}$ in the construction (\ref{anzats}) of quasibosons, for such representation to be self-consistent, are obtained and expression (\ref{gen_solution}) gives their general solution.

In this paper, the novel generalization was carried out, as presented in Section~\ref{sec7}. This is the case of quasibosons made up of two constituents which are  $q$-deformed fermions~(\ref{q-commut})-(\ref{q-nez}). For this generalization, again, we have derived the relations for the defining matrices $\Phi_{\alpha}$ and solved them. Detailed analysis led us at $q\ne 1$ to the resulting structure function (\ref{str_f2}) of the deformed oscillator which provides the exact realization of the quasibosons made up of two $q$-fermions. The principal distinction of the situation treated herein from the case considered in Section~\ref{sec2} (following~\cite{GKM-1}) is such that, while the pure fermions are nilpotent, the $q$-deformed fermions for $q\ne 1$ are not nilpotent of any order, see (\ref{q4}). Since the second order nilpotency of usual fermions (as the no-deformation limit of $q$-fermions) abruptly appears at $q=1$ according to Lemma~\ref{lemma_1}, there is no direct transition from the DSF~(\ref{str_f2}) to DSF~(\ref{solution1}), as a result of the continuous $q\to 1$ limit. See also Remark~\ref{rem3} including~(\ref{gen_DSF}) on this issue.

The general  strategy of the developed approach is to explore  deformed bosons as tools for the realization of quasibosons, which should give considerable simplification (in the algebraic sense) in subsequent applications, achieved when the algebra representing the initial system of composite particles is treated as the algebra corresponding to some deformed oscillator. The obtained results and the developed approach have a potential application to: various problems in (sub)nuclear physics (with such composite particles as hadrons, nucleon complexes) like the study of pairing in
nuclei~\cite{Sviratcheva}; bipartite entangled composites~\cite{GM_Entang} in the Quantum Information Theory (where the role of quasibosons can be played e.g. by biphotons~\cite{Shih}); Bose-Einstein condensation of composite bosons~\cite{Avan2003} and other thermodynamic questions including
the equation of state for many composite bosons systems. Also, the developed formalism can be applied to modeling physical particles or quasiparticles such as excitons, biphonons and cooperons in the corresponding directions of condensed matter physics. Concerning excitons, there already exists~\cite{Combescot_Exc} the description of interacting excitons using infinite series in their creation operators. Besides, excitons were modeled~\cite{Bagheri_Exc,Liu} by $q$-deformed version of bosons, however, without taking into account their compositeness.

As the next steps we intend to study more complicated situations, say, the case of quasibosons composed of two (deformed) bosons, or from two generally deformed fermions. Also, in our nearest plans there is the analysis of composite (quasi-)fermions. Yet another path of extension is to treat quasi-independent  quasibosons, i.e. those with noncommuting different modes.

\ack
This research was  partially supported
by the Special Program of the Division of Physics and Astronomy of NAS of Ukraine.


\appendix
\renewcommand\thesection{Appendix \Alph{section}}

\section{Proof of Proposition~\ref{prop2}}\label{ap1}

As our treatment below concerns only one mode $\alpha$, we will omit the index $\alpha$. Let us first prove the equality (\ref{e1}). For
$n=0$ this is trivial. Put $n=1$:
\begin{equation*}
[A^{\dag}A,A^{\dag}]=A^{\dag}[A,A^{\dag}]=A^{\dag}(1-\Delta_{\alpha\alpha})\equiv
A^{\dag}\varepsilon=A^{\dag}\left[(A^{\dag}A+\varepsilon)^1-(A^{\dag}A)^1\right].
\end{equation*}
Then we proceed by induction. Assuming that the equality holds for $n$, let us prove that it is valid for $n+1$:
\begin{eqnarray*}
\fl\left[(A^{\dag}A)^{n\!+\!1},A^{\dag}\right] \!=\!
\left[A^{\dag}A(A^{\dag}A)^n,A^{\dag}\right] \!=\!
[A^{\dag}A,A^{\dag}](A^{\dag}A)^n \!+\! A^{\dag}A
[(A^{\dag}A)^n,A^{\dag}] \!=\\
=A^{\dag}\varepsilon (A^{\dag}A)^n + A^{\dag}A
A^{\dag}\left[(A^{\dag}A+\varepsilon)^n-(A^{\dag}A)^n\right]
=A^{\dag}\varepsilon (A^{\dag}A)^n +\\
+A^{\dag}(A^{\dag}A\!+\!\varepsilon)^{n\!+\!1} \!-\! A^{\dag}(A^{\dag}A)^{n\!+\!1} \!-\! A^{\dag}\varepsilon
(A^{\dag}A)^n \!=\!
A^{\dag}\left[(A^{\dag}A\!+\!\varepsilon)^{n\!+\!1}\!-\!(A^{\dag}A)^{n\!+\!1}\right].
\end{eqnarray*}
Consider the second equation. When $n=0$ it is also trivial. For $n=1$ we have
\begin{equation*}
[\varepsilon,A^{\dag}] = -fA^{\dag} = A^{\dag}[(-f+\varepsilon)-\varepsilon].
\end{equation*}
The step of induction is:
\[
\fl\left[\varepsilon^{n+1},A^{\dag}\right] =
\left[\varepsilon\varepsilon^n,A^{\dag}\right] =
-fA^{\dag}\varepsilon^n + \varepsilon
A^{\dag}[(-f+\varepsilon)^n-\varepsilon^n] =
\]
\[ = -fA^{\dag}\varepsilon^n
+ (-fA^{\dag}+A^{\dag}\varepsilon)
[(-f+\varepsilon)^n-\varepsilon^n]
=A^{\dag}[(-f+\varepsilon)^{n+1}-\varepsilon^{n+1}].
\]
Thus, the proposition is proven.

\section{Proof of Proposition~\ref{prop3}}\label{ap2}

When $n=0$ the equality reduces to the definition of $N$.  Let us prove it for $n=1$. Present $\chi$ as the formal series:
\begin{equation*}
\chi(x,y)=\sum\limits_{n,m=1}^{\infty}b_{nm} x^ny^m,\quad [x,y]=0.
\end{equation*}
Then
\begin{eqnarray*}
\fl L_1\!=\![\chi(A^{\dag}A,\varepsilon),A^{\dag}]\!=\! \!\sum\limits_{n,m=1}^{\infty}b_{nm}\!\!
\left[(A^{\dag}A)^n,A^{\dag}\right]\varepsilon^m \!+\! \!\sum\limits_{n,m=1}^{\infty}b_{nm}\!
(A^{\dag}A)^n\!\left[\varepsilon^m,A^{\dag}\right] \!= \\
= \!\!\sum\limits_{n,m=1}^{\infty}b_{nm}\!
\left[(A^{\dag}A\!+\!\varepsilon)^n\!-\!(A^{\dag}A)^n\right]\varepsilon^m \!+\! \!\sum\limits_{n,m=1}^{\infty}b_{nm}\!
(A^{\dag}A)^n A^{\dag}[(-\!f\!+\!\varepsilon)^m\!-\!\varepsilon^m]\! =\\
 = A^{\dag}\left[\chi(A^{\dag}A\!+\!\varepsilon, \varepsilon\!-\!f)-\chi(A^{\dag}A,\varepsilon)\right]
=  A^{\dag}\chi(A^{\dag}A\!+\!\varepsilon,\varepsilon\!-\!f) - A^{\dag}N.
\end{eqnarray*}
Next, proceed by induction. The induction step is:
\begin{eqnarray}
\fl L_{n\!+\!1} \!=\! [L_n,A^{\dag}] \!=\! (A^{\dag})^n[\chi(A^{\dag}A\!+\!n\varepsilon\!-\!\sigma_n f, \varepsilon \!\!-\!\! n f),A^{\dag}] \!-\! \!\sum\limits_{k=0}^{n-1}C_n^k (A^{\dag})^{n\!-\!k}[L_k,A^{\dag}]= \nonumber\\
= (A^{\dag})^n[\chi(A^{\dag}A\!+\!n\varepsilon\!-\!\sigma_n f, \varepsilon \!-\! n f),A^{\dag}] \!-\! \sum\limits_{k=0}^{n-1}C_n^k
(A^{\dag})^{n\!-\!k}L_{k+1}.\label{q2}
\end{eqnarray}
Let us transform the commutator in the last expression,
\begin{eqnarray*}
\fl [\chi(A^{\dag}A\!+\!n\varepsilon\!-\!\sigma_n f, \varepsilon \!-\! n f),A^{\dag}] \!=\! \!\sum\limits_{n,m=1}^{\infty}b_{nm} [(A^{\dag}A\!+\!n\varepsilon\!-\!\sigma_n f)^n, A^{\dag}] (\varepsilon \!-\! n f)^m \!+\\
+ \!\!\sum\limits_{n,m=1}^{\infty}b_{nm} (A^{\dag}A\!+\!n\varepsilon\!-\!\sigma_n f)^n[(\varepsilon \!-\! n f)^m, A^{\dag}]\!  =\\
=A^{\dag}\chi\bigl(A^{\dag}A\!+\!(n\!+\!1)\varepsilon\!-\!\sigma_{n+1}f,\varepsilon \!-\!(n\!+\!1)f\bigr) \!-\! A^{\dag}\chi\bigl(A^{\dag}A\!+\!n\varepsilon\!-\!\sigma_n f,\varepsilon \!-\! nf\bigr),
\end{eqnarray*}
where we have used that $\sigma_{n+1}=\sigma_n+n$. Consequently,
\begin{eqnarray}
\fl (A^{\dag})^n[\chi(A^{\dag}A\!+\!n\varepsilon\!-\!\sigma_n f, \varepsilon \!-\! n f),A^{\dag}] \!=\!(A^{\dag})^{n\!+\!1}\chi\bigl(A^{\dag}A \!+\! (n\!+\!1)\varepsilon\!-\!\sigma_{n\!+\!1}f,\varepsilon \!-\!(n\!+\!1)f\bigr) \!-\nonumber\\
-A^{\dag}\cdot (A^{\dag})^n \chi(A^{\dag}A\!+\!n\varepsilon\!-\!\sigma_n f,\varepsilon \!-\! nf).\label{q1}
\end{eqnarray}
Taking into account the induction assumption, the last term takes the form:
\begin{equation*}
A^{\dag}\cdot (A^{\dag})^n
\chi(A^{\dag}A+n\varepsilon-\sigma_n f,\varepsilon - nf) = A^{\dag}L_n +A^{\dag}
\sum\limits_{k=0}^{n-1}C_n^k (A^{\dag})^{n-k}L_k.
\end{equation*}
Substituting this expression into (\ref{q1}) and then the resulting expression into (\ref{q2}), we obtain
\begin{eqnarray*}
\fl L_{n\!+\!1} \!\!=\! (A^{\dag})^{n\!+\!1} \!\chi\bigl(\!A^{\dag}A\!\!+\!\!(n\!\!+\!\!1)\varepsilon\!\!-\!\!\sigma_{n\!+\!1}\!f,\varepsilon \!\!-\!\!(n\!\!+\!\!1)\!f\!\bigr) \!-\!\! A^{\dag}L_n \!\!-\!\! A^{\dag} \!\sum\limits_{k=0}^{n-1}\!C_n^k (A^{\dag})^{n\!-\!k}\!L_k \!\!-\!\! \!\sum\limits_{k=0}^{n-1}\!C_n^k (A^{\dag})^{n\!-\!k}\!L_{k\!+\!1}\\
=(A^{\dag})^{n\!+\!1}\chi\bigl(A^{\dag}A\!+\!(n\!+\!1)\varepsilon\!-\!\sigma_{n\!+\!1}f,\varepsilon \!-\!(n\!+\!1)f\bigr) - \sum\limits_{k=0}^{n}C_{n\!+\!1}^k (A^{\dag})^{n\!+\!1\!-\!k}L_k.
\end{eqnarray*}
Thus, the proposition is proven.

\section[Particular examples]{Particular examples: $\mu,\nu=1$ and $\mu,\nu=\overline{1,2}$}\label{ap3}

For both the examples we assume $q\ne 1$.

\noindent{\bf Example 1.} Let us consider first the simplest case when only one mode is possible for a composite boson's constituents: $\mu,\nu=1$. The number of the modes $\alpha$ is not significant here, as further treatment concerns only one fixed arbitrary mode $\alpha$. Then the creation and annihilation operators $A^{\dag}_{\alpha}$, $A_{\alpha}$ of the composite boson according to (\ref{anzats}) reduce to (the fixed index $\alpha$ is omitted)
\begin{equation}
A^{\dag} = \Phi^{11}a^{\dag}_1b^{\dag}_1,\quad A = \overline{\Phi^{11}} b_1 a_1.
\end{equation}
We require the composite bosons to be algebraically represented by deformed bosons, i.e. that defining equality $\mathcal{A}\mathcal{A}^\dag=\phi(\mathcal{N}+1)$ for deformed bosons holds on any $n$-composite bosons state $(A^{\dag})^n|O\rangle$:
\begin{eqnarray}
A A^\dag\cdot (A^\dag)^n|O\rangle=\phi(N+1)(A^{\dag})^n|O\rangle\ \Leftrightarrow\nonumber\\
|\Phi^{11}|^2 [n^a_1+1]_{-q}[n^b_1+1]_{-q}(A^\dag)^n|O\rangle = \phi(n+1)(A^\dag)^n|O\rangle,\label{n=1_req}
\end{eqnarray}
where we have introduced the number operators $n^a_1$, $n^b_1$ for $q$-deformed constituent fermions.   Taking into account the normalization condition $|\Phi^{11}|^2=1$ and the equality $n^i_1 (A^\dag)^n|O\rangle = n (A^\dag)^n|O\rangle$, $i=a,b$, relation~(\ref{n=1_req}) is rewritten as
\begin{equation*}
[n+1]_{-q}[n+1]_{-q}(A^\dag)^n|O\rangle = \phi(n+1)(A^\dag)^n|O\rangle.
\end{equation*}
The latter implies $\phi(n)= ([n]_{-q})^2$.

\noindent{\bf Example 2.} Next, let us consider the case of two modes $\mu=1,2$ and $\nu=1,2$. For the creation and annihilation operators $A^{\dag}_{\alpha}$, $A_{\alpha}$ we obtain (the fixed $\alpha$ is omitted as before)
\begin{eqnarray}
A^{\dag} = \sum_{\mu,\nu=1}^2 \Phi^{\mu\nu}a_{\mu}^{\dag}b_{\nu}^{\dag} = \Phi^{11}a^{\dag}_1b^{\dag}_1 + \Phi^{12}a^{\dag}_1b^{\dag}_2 + \Phi^{21}a^{\dag}_2b^{\dag}_1 + \Phi^{22}a^{\dag}_2b^{\dag}_2,\\
A = \sum_{\mu,\nu=1}^2 \overline{\Phi^{\mu\nu}}b_{\nu}a_{\mu} = \overline{\Phi^{11}}b_1a_1 + \overline{\Phi^{12}}b_2a_1 + \overline{\Phi^{21}}b_1a_2 + \overline{\Phi^{22}}b_2a_2.
\end{eqnarray}
As in Example~1, for the validity of the realization of  composite bosons by deformed bosons we require the following equality
\begin{equation}
A (A^\dag)^{n+1}|O\rangle=\phi(N+1)(A^{\dag})^n|O\rangle.\label{2_req}
\end{equation}
As an auxiliary step, let us present the operator $(A^\dag)^n$ as the sum
\begin{equation}
(A^\dag)^n = \sum_{k,l=0}^n (-1)^{n(n-1)/2}C_n^{kl}(\Phi) (a_2^\dag)^k(a_1^\dag)^{n-k}(b_2^\dag)^l(b_1^\dag)^{n-l}
\end{equation}
with the coefficients $C_n^{kl}(\Phi)\equiv C_n^{kl}(\Phi^{11},\Phi^{12},\Phi^{21},\Phi^{22})$ written in the form
\begin{equation}
C_n^{kl}(\Phi) = \sum_{j=\max(0,l+k-n)}^{\min(k,l)} P_n^{kl}(j) (\Phi^{22})^j(\Phi^{21})^{k-j}(\Phi^{12})^{l-j}(\Phi^{11})^{n-k-l+j}.\label{C_def}
\end{equation}
After some algebra we derive the following recurrence relations for the coefficients $P_n^{kl}(j)$:
\begin{eqnarray*}
\fl P^{kl}_{n+1}(j) = P^{kl}_n(j) + (-1)^{n+k-1}P_n^{k-1,l}(j) + (-1)^{n+l-1}P_n^{k,l-1}(j) + (-1)^{k+l}P_n^{k-1,l-1}(j-1),\\
P^{0l}_{n+1}(0) = P_n^{0l}(0) + (-1)^{n+l-1}P_n^{0,l-1}(0),\\
P^{k0}_{n+1}(0) = P_n^{k0}(0) + (-1)^{n+k-1}P_n^{k-1,0}(0),\\
P^{n+1,l}_{n+1}(l) = P_n^{nl}(l) + (-1)^{n+l-1} P^{n,l-1}_n(l-1),\\
P^{k,n+1}_{n+1}(k) = P_n^{kn}(k) + (-1)^{n+k-1} P^{k-1,n}_n(k-1),\quad 1\le k,l \le n,
\end{eqnarray*}
with the initial conditions
\begin{eqnarray*}
P_n^{00}(0) = P_n^{0n}(0) = P_n^{n0}(0) = P_n^{nn}(n) = 1,\\
P_n^{kl}(j) = 0\ \ {\rm if}\ \ j>\min(k,l)\ \ {\rm or} \ \ j<\max(0,k+l-n).
\end{eqnarray*}
In what follows we need only a few coefficients $P^{kl}_n(j)$ for which we give their explicit expressions, as the solutions of the above recurrence relations:
\begin{eqnarray}
P_n^{01}(0) = P_n^{10}(0) = \frac{1-(-1)^n}{2},\label{coef_1}\\
P_n^{11}(0) = -n + \frac{1-(-1)^n}{2},\quad P_n^{11}(1) = n,\label{coef_2}\\
P_n^{02}(0) = P_n^{20}(0) = \frac12 n + \frac{(-1)^n-1}{4},\label{coef_3}\\
P_{n+1}^{12}(0) = P_{n+1}^{21}(0) = \frac32 n\frac{1+(-1)^n}{2},\label{coef_4}\\
P_{n+1}^{12}(1) = P_{n+1}^{21}(1) = -n\frac{1+(-1)^n}{2},\label{coef_5}\\
P_{n+1}^{22}(0) = (3/4-3n/2)\frac{1+(-1)^n}{2} + \frac34 n^2-\frac34,\label{coef_6}\\
P_{n+1}^{22}(1) = (n-1)\frac{1+(-1)^n}{2}-n^2+1,\label{coef_7}\\
P_{n+1}^{22}(2) = n(n+1)/2.\label{coef_8}
\end{eqnarray}
Now rewrite the l.h.s. and r.h.s. of (\ref{2_req}) respectively as
\begin{eqnarray*}
\fl AA^{n+1}|O\rangle = \sum_{k,l=0}^n (-1)^{n(n-1)/2} \Bigl\{[k+1]_{-q}[l+1]_{-q}\overline{\Phi^{22}} C_{n+1}^{k+1,l+1}(\Phi) +\\
\fl\qquad+ (-1)^l [k\!+\!1]_{-q}[n\!+\!1\!-\!l]_{-q} \overline{\Phi^{21}} C_{n+1}^{k+1,l}(\Phi) + (-1)^k [n\!+\!1\!-\!k]_{-q}[l\!+\!1]_{-q} \overline{\Phi^{12}} C_{n+1}^{k,l+1}(\Phi)\!+\\
+ (-1)^{k+l} [n\!+\!1\!-\!k]_{-q}[n\!+\!1\!-\!l]_{-q} \overline{\Phi^{11}} C_{n+1}^{kl}(\Phi)\Bigr\} (a_2^\dag)^k (a_1^\dag)^{n-k} (b_2^\dag)^l (b_1^\dag)^{n-l}|O\rangle,\\
\fl \phi(N+1)(A^{\dag})^n|O\rangle = \phi(n+1) \sum_{k,l=0}^n (-1)^{n(n-1)/2}C_n^{kl}(\Phi) (a_2^\dag)^k(a_1^\dag)^{n-k}(b_2^\dag)^l(b_1^\dag)^{n-l}|O\rangle.
\end{eqnarray*}
Since the vectors $(a_2^\dag)^k(a_1^\dag)^{n-k}(b_2^\dag)^l(b_1^\dag)^{n-l}|O\rangle$, $k,l = 0,...,n$, $n=0,1,2,...$, form a basis in the Hilbert space, we may equate the corresponding summands, and thus the requirement (\ref{2_req}) for $n\ge1$ is equivalent to the following system of equations:
\begin{eqnarray}
\fl [k+1]_{-q}[l+1]_{-q}\overline{\Phi^{22}} C_{n+1}^{k+1,l+1}(\Phi)
+ (-1)^l [k+1]_{-q}[n+1-l]_{-q} \overline{\Phi^{21}} C_{n+1}^{k+1,l}(\Phi) +\nonumber\\
\fl + (-1)^k [n\!+\!1\!-\!k]_{-q}[l\!+\!1]_{-q} \overline{\Phi^{12}} C_{n+1}^{k,l+1}(\Phi)
+ (-1)^{k+l} [n\!+\!1\!-\!k]_{-q}[n\!+\!1\!-\!l]_{-q} \overline{\Phi^{11}} C_{n+1}^{kl}(\Phi) -\nonumber\\
- \phi(n+1) C_n^{kl}(\Phi) = 0,\qquad k,l = 0,...,n,\  n=1,2,...\  .\label{2_system}
\end{eqnarray}
Taking here $k=l=0$, and using (\ref{C_def}) and expressions (\ref{coef_1})-(\ref{coef_2}), we arrive at the equations
\begin{eqnarray*}
\fl\Bigl(-n+\frac{1+(-1)^n}{2}-1\Bigr)\overline{\Phi^{22}}\Phi^{21}\Phi^{12}(\Phi^{11})^{n-1} + (n+1)\overline{\Phi^{22}} \Phi^{22} (\Phi^{11})^n +\\
+ \frac{1+(-1)^n}{2}[n+1]_{-q}\overline{\Phi^{21}}\Phi^{21}(\Phi^{11})^n + \frac{1+(-1)^n}{2}[n+1]_{-q} \overline{\Phi^{12}}\Phi^{12}(\Phi^{11})^n +\\
+ [n+1]_{-q}^2\overline{\Phi^{11}}(\Phi^{11})^{n+1} = \phi(n+1) (\Phi^{11})^n,\quad n=1,2,...\,.
\end{eqnarray*}
Due to the normalization condition the matrix $\Phi$ has at least one nonzero element.  Let this be $\Phi^{11}\ne 0$ (for any other nonzero element the subsequent treatment is analogous). Then, from the last equation after dividing by $(\Phi^{11})^n$ and replacing $n+1\rightarrow n$ we obtain the following expression for the structure function:
\begin{eqnarray}
\fl \phi(n) = \Bigl(\frac{1-(-1)^n}{2}-n\Bigr)\overline{\Phi^{22}}\Phi^{21}\Phi^{12}(\Phi^{11})^{-1} + n\overline{\Phi^{22}} \Phi^{22} + \frac{1-(-1)^n}{2}[n]_{-q}\overline{\Phi^{21}}\Phi^{21} +\nonumber\\
+\frac{1-(-1)^n}{2}[n]_{-q} \overline{\Phi^{12}}\Phi^{12} + [n]_{-q}^2\overline{\Phi^{11}}\Phi^{11}.\label{2_str_fun}
\end{eqnarray}
Next, let us take $k=1$, $l=0$ in (\ref{2_system}).  Then, using (\ref{C_def}) together with (\ref{coef_1})-(\ref{coef_5}), after respective transformations we rewrite (\ref{2_system})  in the following form:
\begin{eqnarray}
\fl f_1(n) \overline{\Phi^{22}}\Phi^{12}(\Phi^{21})^2 + f_2(n) \overline{\Phi^{22}}\Phi^{22}\Phi^{21}\Phi^{11} + f_3(n) \overline{\Phi^{12}}\Phi^{12}\Phi^{21}\Phi^{11} + f_4(n) \overline{\Phi^{12}}\Phi^{22}(\Phi^{11})^2 +\nonumber\\
+ f_5(n) \overline{\Phi^{21}}(\Phi^{21})^2\Phi^{11} + f_6(n) \overline{\Phi^{11}}\Phi^{21}(\Phi^{11})^2 = 0,\label{2_eq2}
\end{eqnarray}
where
\begin{eqnarray*}
\fl\qquad f_1(n) = [2]_{-q} P^{21}_{n+1}(0) - P^{10}_{n}(0)P^{11}_{n+1}(0) =  [2]_{-q}\frac{3n}{2}\frac{1\!+\!(-1)^n}{2} + (n\!+\!1)\frac{1\!-\!(-1)^n}{2},\\
\fl\qquad f_2(n) = [2]_{-q}P^{21}_{n+1}(1) - P^{10}_n(0)P^{11}_{n+1}(1) = -[2]_{-q} n \frac{1\!+\!(-1)^n}{2} - (n\!+\!1)\frac{1\!-\!(-1)^n}{2},\\
\fl\qquad f_3(n) = -[n]_{-q}P^{11}_{n+1}(0) - [n\!+\!1]_{-q}P^{10}_n(0)P^{01}_{n+1}(0) = [n]_{-q}\bigl(n+\frac{1-(-1)^n}{2}\bigr),\\
\fl\qquad f_4(n) = -[n]_{-q} P^{11}_{n+1}(1) = -[n]_{-q}(n+1),\\
\fl\qquad f_5(n) = [2]_{-q}[n\!+\!1]_{-q}P^{20}_{n+1}(0) \!-\! [n\!+\!1]_{-q}P^{10}_n(0)P^{10}_{n+1}(0) = [2]_{-q}[n\!+\!1]_{-q}\Bigl(\frac{n}{2}\!+\!\frac{1\!-\!(-1)^n}{4}\Bigr),\\
\fl\qquad f_6(n) = -[n]_{-q}[n\!+\!1]_{-q}P^{10}_{n+1}(0) - [n\!+\!1]^2_{-q}P^{10}_n(0) =\\
= -[n]_{-q}[n\!+\!1]_{-q} \frac{1\!+\!(-1)^n}{2} - [n\!+\!1]_{-q}^2 \frac{1\!-\!(-1)^n}{2}.
\end{eqnarray*}
Using the linear independence of $f_6(n)$ from $f_1(n),...,f_5(n)$ and that $\Phi^{11}\ne 0$, from (\ref{2_eq2}) we obtain
\begin{equation}
\overline{\Phi^{11}}\Phi^{21}(\Phi^{11})^2 f_6(n) = 0
\quad \Rightarrow \quad \Phi^{21} = 0. \label{el2}
\end{equation}
Likewise, considering the case of $k=0$, $l=1$ in (\ref{2_system}) we deduce
\begin{equation}
\Phi^{12} = 0.\label{el3}
\end{equation}
Now examine the case $k=l=1$ in (\ref{2_system}). Substituting (\ref{C_def}) together with (\ref{coef_2})-(\ref{coef_8}) into (\ref{2_system}), after dividing by $(\Phi^{22})^{n-3}$ we obtain
\begin{eqnarray}
\fl g_1(n) \overline{\Phi^{22}}(\Phi^{21})^2(\Phi^{12})^2 \!+\! g_2(n) \overline{\Phi^{22}}\Phi^{22}\Phi^{21}\Phi^{12}\Phi^{11} \!+\! g_3(n) \overline{\Phi^{22}}(\Phi^{22})^2(\Phi^{11})^2 \!+\nonumber\\
\fl\qquad + g_4(n) \overline{\Phi^{21}}(\Phi^{21})^2\Phi^{12}\Phi^{11} \!+\! g_5(n) \overline{\Phi^{21}}\Phi^{22}\Phi^{21}(\Phi^{11})^2 \!+\! g_6(n) \overline{\Phi^{12}}\Phi^{21}(\Phi^{12})^2\Phi^{11} \!+\label{2_eq3}\\
\fl\qquad\qquad + g_7(n)\overline{\Phi^{12}}\Phi^{22}\Phi^{12}(\Phi^{11})^2 \!+\! g_8(n) \overline{\Phi^{11}}\Phi^{21}\Phi^{12}(\Phi^{11})^2 \!+\! g_9(n) \overline{\Phi^{11}}\Phi^{22}(\Phi^{11})^3 \!=\! 0,\nonumber
\end{eqnarray}
where
\begin{eqnarray*}
g_1(n) = [2]_{-q}^2 P^{22}_{n+1}(0) \!-\! P^{11}_{n+1}(0)P^{11}_n(0),\\
g_2(n) = [2]_{-q}^2 P^{22}_{n+1}(0) \!-\! P^{11}_{n+1}(0)P^{11}_n(1) \!-\! P^{11}_{n+1}(1)P^{11}_n(0),\\
g_3(n) = [2]_{-q}^2 P^{22}_{n+1}(2) \!-\! P^{11}_{n+1}(1)P^{11}_n(1) = ([2]_{-q}^2/2-1) n(n\!+\!1),\\
g_4(n) = g_6(n) = - [2]_{-q}[n]_{-q}P^{21}_{n+1}(0) \!-\! [n\!+\!1]_{-q}P^{10}_{n+1}(0)P^{11}_n(0),\\
g_5(n) = g_7(n) = - [2]_{-q}[n]_{-q}P^{21}_{n+1}(1) \!-\! [n\!+\!1]_{-q}P^{10}_{n+1}(0)P^{11}_n(1),\\
g_8(n) = [n]_{-q}^2P^{11}_{n+1}(0) \!-\! [n\!+\!1]_{-q}^2P^{11}_n(0),\\
g_9(n) = [n]_{-q}^2P^{11}_{n+1}(1) \!-\! [n\!+\!1]_{-q}^2P^{11}_n(1) = [n]_{-q}^2(n\!+\!1) \!-\! [n\!+\!1]_{-q}^2n.
\end{eqnarray*}
Taking into account (\ref{el2}) and (\ref{el3}) relation~(\ref{2_eq3}) reduces to
\begin{equation}
g_3(n) \overline{\Phi^{22}}(\Phi^{22})^2(\Phi^{11})^2 + g_9(n) \overline{\Phi^{11}}\Phi^{22}(\Phi^{11})^3 = 0.\label{2_eq4}
\end{equation}
Using again the linear independence of the functions $g_6(n)$ and $g_9(n)$, and recalling that $\Phi^{11}\ne 0$, from (\ref{2_eq4}) we obtain
\begin{equation*}
\overline{\Phi^{11}}\Phi^{22}(\Phi^{11})^3 = 0 \quad \Rightarrow \quad \Phi^{22}=0.
\end{equation*}
Substitution of this element along with two previous ones (\ref{el2}) and (\ref{el3}) in (\ref{2_str_fun}), and the use of the normalization condition gives the resulting expression for the structure function in the case under consideration: \ $\phi(n) = [n]_{-q}^2$.

Remark that if $q=0$, the uncertainty ``$0^0$'' in $\phi(n)$ at $n=0$ is resolved using the condition $\phi(0)=0$ that results in~$\phi(n)=\theta(n)$.


\section*{References}

\begin{thebibliography}{10}

\bibitem{Bona} D. Bonatsos and C. Daskaloyannis,              
    \link{http://dx.doi.org/10.1016/S0146-6410(99)00100-3}{{\it Prog. Part. Nucl. Phys.} {\bf 43}, 537 (1999)}.
\bibitem{Per} W.R. Perkins,                                  
    \link{http://dx.doi.org/10.1023/A:1015728722664}{{\it Int. J. Th. Phys.} {\bf 41}, 823 (2002)}.
\bibitem{Avan} S.S. Avancini and G. Krein,                  
    \link{http://dx.doi.org/10.1088/0305-4470/28/3/021}{{\it J. Phys. A: Math. Gen.} {\bf 28}, 685 (1995)}.
\bibitem{Green} O.W. Greenberg,                             
    \link{http://dx.doi.org/10.1103/PhysRevD.43.4111}{{\it Phys. Rev. D} {\bf 43}, 4111 (1991)}.
\bibitem{Arik} M. Arik and D.D. Coon,                       
    \link{http://dx.doi.org/10.1063/1.522937}{{\it J. Math. Phys.} {\bf 17}, 524 (1976)}.
\bibitem{Jag-etal} R Jagannathan {\it et al},
    \link{http://dx.doi.org/10.1088/0305-4470/25/23/036}{{\it J. Phys. A: Math. Gen.} {\bf 25}, 6429 (1992)}.
\bibitem{Biedenharn} L.C. Biedenharn,                               
    \link{http://dx.doi.org/10.1088/0305-4470/22/18/004}{{\it J. Phys. A: Math. Gen.} {\bf 22}, L873 (1989)}.
\bibitem{Macfarlane}  A.J. Macfarlane,                                           
    \link{http://dx.doi.org/10.1088/0305-4470/22/21/020}{{\it J. Phys.A: Gen.} {\bf 22}, 4581 (1989)}.
\bibitem{TD} K. Odaka, T. Kishi and S. Kamefuchi,           
    \link{http://dx.doi.org/10.1088/0305-4470/24/11/004}{{\it J. Phys. A: Math. Gen.} {\bf 24}, L591 (1991)}.
\bibitem{TD2} S. Chaturvedi, V. Srinivasan and R. Jagannathan,        
    \link{http://dx.doi.org/10.1142/S0217732393003457}{{\it Mod. Phys. Lett. A} {\bf 8}, 3727 (1993)}.
\bibitem{GR1} A.M. Gavrilik and A.P. Rebesh,                
    \link{http://dx.doi.org/10.1142/S0217732307022827}{{\it Mod. Phys. Lett. A} {\bf 22}, 949 (2007)}.
\bibitem{p-q} A. Chakrabarti and R. Jagannathan,                       
    \link{http://dx.doi.org/10.1088/0305-4470/24/13/002}{{\it J. Phys.A: Math. Gen.} {\bf 24}, L711 (1991)}.
\bibitem{Arik-Fibo} M. Arik et al.                          
    \link{http://dx.doi.org/10.1007/BF01558292}{{\it Z. Phys. C} {\bf 55}: 89-95 (1992)}.
\bibitem{Jan} A. Jannussis,                                 
    \link{http://dx.doi.org/10.1088/0305-4470/26/5/011}{{\it J. Phys. A: Math. Gen.} {\bf 26}, L233 (1993)}.
\bibitem{GKR} A.M. Gavrilik, I.I Kachurik and A.P. Rebesh   
    \link{http://dx.doi.org/10.1088/1751-8113/43/24/245204}{{\it J. Phys. A: Math.Theor.} {\bf 43}, 245204 (2010)}.
\bibitem{GR2} A.M. Gavrilik and A.P. Rebesh
    \link{http://dx.doi.org/10.1140/epja/i2011-11055-x}{{\it Eur. Phys. J. A} {\bf 47}, 55 (2011)}.
\bibitem{Melja} S. Meljanac, M. Milekovi\'{c} and S. Pallua, 
    \link{http://dx.doi.org/10.1016/0370-2693(94)90427-8}{{\it Phys. Lett. B} {\bf 328}, 55 (1994)}.
\bibitem{Man'ko} V.I. Man'ko et al.,                          
    \link{http://dx.doi.org/10.1088/0031-8949/55/5/004}{{\it Phys. Scripta} {\bf 55}, 528 (1997)}.
\bibitem{Rego2} M. Rego-Monteiro, L.M.C.S. Rodrigues and S. Wulck,         
    \link{http://dx.doi.org/10.1103/PhysRevLett.76.1098}{{\it Phys. Rev. Lett.} {\bf 76}, 1098 (1998)}.
\bibitem{Rego3} M. Rego-Monteiro, L.M.C.S. Rodrigues and S. Wulck,
    \link{http://dx.doi.org/10.1016/S0378-4371(97)00633-X}{{\it Physica A} {\bf 259}, 245 (1998)}.
\bibitem{AGI1} D. Anchishkin, A. Gavrilik and N. Iorgov,      
    \link{http://dx.doi.org/10.1007/PL00013603}{{\it Eur. Phys. J. A} {\bf 7}, 229 (2000)}.
\bibitem{AGI1-2} A.M. Gavrilik,
    \link{http://dx.doi.org/10.1016/S0920-5632(01)01570-5}{{\it Nucl. Phys. B Proc. Suppl.} {\bf 102}, 298 (2001)}.
\bibitem{AGI2} D. Anchishkin, A. Gavrilik and N. Iorgov,      
    \link{http://dx.doi.org/10.1142/S0217732300001754}{{\it Mod. Phys. Lett. A} {\bf 15}, 1637 (2000)}.
\bibitem{Avan2003} S.S. Avancini, J.R. Marinelli and G. Krein,
    \link{http://dx.doi.org/10.1088/0305-4470/36/34/307}{{\it J. Phys. A: Math. Gen.} {\bf 36}, 9045 (2003)}.
\bibitem{AG} L.V. Adamska and A.M. Gavrilik,                  
    \link{http://dx.doi.org/10.1088/0305-4470/37/17/009}{{\it J. Phys. A: Math. Gen.} {\bf 37}, 4787 (2004)}.
\bibitem{gavrS} A. M. Gavrilik,                               
    \link{http://dx.doi.org/10.3842/SIGMA.2006.074}{{\it SIGMA} {\bf 2}, Paper 074, 12 pages (2006)}.
\bibitem{Rego1} C.I. Ribeiro-Silva, E.M.F. Curado and M.A. Rego-Monteiro,
    \link{http://dx.doi.org/10.1088/1751-8113/41/14/145404}{{\it J. Phys. A: Math. Theor.} {\bf 41},} 
\link{http://dx.doi.org/10.1088/1751-8113/41/14/145404}{ 145404 (2008)}.
\bibitem{GKM-1} A.M. Gavrilik, I.I. Kachurik and Yu.A. Mishchenko,            
    \link{http://www.ujp.bitp.kiev.ua/files/file/papers/56/9/560911p.pdf}{{\it Ukr. J. Phys.} {\bf 56}, 948 (2011)}.
\bibitem{GR4} A.M. Gavrilik and A.P. Rebesh,                
    \link{http://dx.doi.org/10.1142/S021773230802687X}{{\it Mod. Phys. Let. A} {\bf 23}, 921 (2008)}.
\bibitem{GR5} A.M. Gavrilik, and A.P. Rebesh,                
    \link{http://dx.doi.org/10.1088/1751-8113/43/9/095203}{{\it J. Phys. A: Math.Theor.} {\bf 43}, 095203 (2010)}.
\bibitem{Korn} G.A. Korn and T.M. Korn,                       
    {\it Mathematical Handbook for Scientists and Engineers}, McGraw-Hill Companies, 1967.
\bibitem{Viswanthan} K.S. Viswanathan {\it et al.},
    \link{http://dx.doi.org/10.1088/0305-4470/25/7/009}{{\it J. Phys. A: Math. Gen.} {\bf 25}, L335 (1992)}.
\bibitem{Sviratcheva} K. D. Sviratcheva {\it et. al.},
    \link{http://dx.doi.org/10.1103/PhysRevLett.93.152501}{{\it Phys. Rev. Lett.} {\bf 93}, 152501 (2004)}.
\bibitem{GM_Entang} A.M. Gavrilik and Yu.A. Mishchenko,
    \link{http://arxiv.org/abs/1108.0936}{arXiv:1108.0936}.
\bibitem{Shih} Y. Shih,
    \link{http://dx.doi.org/10.1088/0034-4885/66/6/203}{{\it Rep. Prog. Phys.} {\bf 66}, 1009 (2003)}.
\bibitem{Combescot_Exc} M. Combescot and O. Betbeder-Matibet,
    \link{http://dx.doi.org/10.1103/PhysRevB.78.125206}{{\it Phys. Rev. B} {\bf 78}, 125206 (2008)}.
\bibitem{Bagheri_Exc} M. Bagheri Harouni, R. Roknizadeh and M. H. Naderi,
    \link{http://dx.doi.org/10.1088/0953-4075/42/9/095501}{{\it J. Phys. B: At. Mol. Opt. Phys.} {\bf 42},}
\link{http://dx.doi.org/10.1088/0953-4075/42/9/095501}{ 095501 (2009)}.
\bibitem{Liu} Yu-Xi Liu {\it et al.},
    \link{http://dx.doi.org/10.1103/PhysRevA.63.023802}{{\it Phys. Rev. A} {\bf 63}, 023802 (2001)}.

\end{thebibliography}
\end{document}